\documentclass[preprint]{article}
\usepackage{graphicx}
\usepackage{amsmath}
\usepackage{array}
\usepackage{arydshln}
\usepackage{amssymb}
\usepackage{rotating}
\usepackage[hidelinks]{hyperref}
 \usepackage{amsthm}
 \usepackage{subcaption}
 \usepackage{float}

\usepackage{stmaryrd}

\usepackage{booktabs}
\usepackage{collcell}
\newcommand*{\movedown}[1]{%
  \smash{\raisebox{-1ex}{#1}}}
\newcolumntype{q}{>{\collectcell\movedown}r<{\endcollectcell}}

\usepackage{array}

 \usepackage{xcolor}
\newtheorem{thm}{Theorem}[section]
\newtheorem{cor}{Corollary}[section]
\newtheorem{prop}{Proposition}[section]
\newtheorem{conj}{Conjecture}[section]
\newcommand*{\img}[1]{%
    \raisebox{-.3\baselineskip}{%
        \includegraphics[
        height=\baselineskip,
        width=\baselineskip,
        keepaspectratio,
        ]{#1}%
    }%
}
\usepackage{wrapfig}
\graphicspath{ {./images/} }

\newcommand{\be}{\begin{equation}}
\newcommand{\ee}{\end{equation}}
\newcommand{\ba}{\begin{eqnarray}}
\newcommand{\ea}{\end{eqnarray}}

\newcommand{\al}{\alpha}

\begin{document}
\setlength\dashlinedash{.6pt}
\setlength\dashlinedash{2pt}
\hoffset=-.4truein\voffset=-0.5truein
\setlength{\textheight}{8.5 in}

\begin{titlepage}
\begin{center}
\hfill { 
}\\
\vskip 20 mm
{\large \bf   
A relation between the HOMFLY--PT and Kauffman polynomials via characters}

\vskip 10mm
{\bf {Andreani Petrou and Shinobu Hikami }}
\vskip 5mm

Okinawa Institute of Science and Technology Graduate University,\\
 1919-1 Tancha, Okinawa 904-0495, Japan.

\vskip 5mm
{\bf Abstract}
\vskip 3mm

\end{center}
The HOMFLY--PT and Kauffman polynomials are related to each other for special classes of knots constructed by full twists and Jucys-Murphy twists. The conditions for this relation are articulated in terms of characters of the Birman-Murakami-Wenzl algebra. The latter are the coefficients in the expansion of the Kauffman polynomial involving the quantum dimensions of $SO(N+1)$.
This expansion  
allows to prove the conjectural 1-1 correspondence between the HOMFLY--PT/Kauffman relation and the Harer-Zagier (HZ) factorisability for a large family of 3-strand knots. The conjecture  remains open for knots with 4 and higher strands. 

 \end{titlepage}
 \tableofcontents

\section{Introduction}
 The HOMFLY--PT  and Kauffman polynomials, as 
 two-variable generalisations of the celebrated Jones polynomial, have attracted a lot of attention by both mathematicians and physicists. They are both strong knot invariants, which,  in general,  can distinguish different sets of knots and links. 
 In the context of TQFT, these correspond  to  the  Lie groups 
 $SU(N)$ and $SO(N+1)$, respectively. From this perspective, it becomes obvious that they coincide at $N=2$,  corresponding to  the Jones polynomial, due to the isomorphism $\mathfrak{su}(2)\simeq \mathfrak{so}(3)$. Articulating a general relation between these polynomials for arbitrary $N$ is a non-trivial task, however, and it can only be achieved for special classes of knots and links, as explained below.

A  relation between the HOMFLY--PT and Kauffman polynomials  was first found for torus knots by Labastida and Pérez  \cite{Perez}. Explicitly
 \ba\label{KFhatINTRO}
&&\hspace{-9mm}H(A,z)=KF_{\rm even}(iA,iz)-\frac{z}{A-A^{-1}} KF_{\rm odd}(iA,iz)=:\widehat{KF}(A,z), 
 \ea  
where the subscripts  refer to the even and odd powers of $A$ and $z$.
 Finding more knots, beyond the torus family that satisfy this relation has  been an open problem for nearly three decades. However, a solution has been recently provided in \cite{Petrou2}, where  relation (\ref{KFhatINTRO}) was proven to hold for  several infinite families of hyperbolic knots. 
In Theorem~\ref{thm:Kjkl,H-KF} of the present article, we extend this result by proving that the HOMFLY--PT/Kauffman relation (\ref{KFhatINTRO}) holds for a more general family of 3-strand knots, which, as argued in \cite{Petrou3}, can be thought of as  a hyperbolic extension of torus knots. 

 From a physics perspective, the HOMFLY--PT/Kauffman relation has a peculiar implication for the BPS states of topological  strings \cite{Bouchard,Mironov}. In particular, the difference between the $SU(N)$ and $SO(N+1)$ polynomials, corresponds to the difference between oriented and unoriented surfaces, and hence it is described as the sum of the 1- and 2-cross cap BPS  invariants, as
discussed in \cite{Petrou2}. When $(\ref{KFhatINTRO})$ is true, the 2-cross cap BPS invariants vanish \cite{Petrou2}.

 The hyperbolic families that satisfy the HOMFLY--PT/Kauffman relation were discovered through the Harer-Zagier (HZ) transform  of the
  HOMFLY--PT polynomial, which  has been  extensively investigated
  in the previous articles \cite{Petrou1,Petrou2,Petrou3}. 
  The HZ transform  is a discrete Laplace transform that maps the HOMFLY--PT polynomial into a rational function.
  This function is said to be \emph{factorisable} when both the numerator and the denominator become factorised into a product of monomials. 
HZ factorisability, which occurs for torus knots, motivated the construction of several families of hyperbolic  knots, which are generated by combinations of full twists and Jucys-Murphy twists, as described in \cite{Petrou2}. In the 3-strand case, these families are precisely the ones that are proven to satisfy the HOMFLY--PT/Kauffman relation (\ref{KFhatINTRO}), in support of the conjecture  \cite{Petrou2}
\vspace{-2mm}
\be\label{conjINTRO}
\text{HOMFLY--PT/Kauffman relation $\iff$ HZ factorisability.}
\ee

A revealing approach that enables a deeper understanding of  (\ref{KFhatINTRO}) and (\ref{conjINTRO}) comes from representation theory.
In the previous article \cite{Petrou3}, the $SU(N)$  character expansion \cite{Morozov,Petrou3} has provided insights into the HOMFLY--PT polynomial and the factorisability properties of its HZ transform. Explicitly, the $SU(N)$ expansion for a knot with an $m$ strand braid representative reads 
\be\label{XSUINTRO}
H = A^{-w}\sum_{Q} h^{Q} \frac{S_{Q}}{S_{[1]}}=: {\mathbb{X}}^{SU},\ee
where the sum is over Young diagrams $Q$ with $m$ boxes, $S_Q$ are the Schur functions and $h^Q$ are the Racah coefficients  ($6j$ symbols), determined by traces of $R$-matrices. 
Conditions for HZ factorisability for each fixed number of strands $m$, can be expressed in terms of the coefficients $h^Q$. 
In the present paper, aiming to gain deeper insights into the HOMFLY--PT/Kauffman relation, we consider a similar expansion for the Kauffman polynomial. 

In Sec.~\ref{sec:dubrovnik}, we take a naive approach, in which we difine ${\mathbb{X}}^{SO}$ similar to (\ref{XSUINTRO}), with only the Schur functions $S_Q$ being replaced by the quantum dimensions $d_{Q}$ of $SO(N+1)$. 
A derivation of $d_Q$ from the classical formula of the $SO(N+1)$ Schur functions is included in Appendix A. 
It turns out that for torus knots $T(m,n)$, with  $m$ odd, ${\mathbb{X}}^{SO}$ yields precisely the Dubrovnik version of the Kauffman polynomial $Y(A,q)$   \cite{Perez,Stevan}. However, more generally holds  
\be\label{Delta}
 Y(A,q)  = {\mathbb{X}}^{SO}(A,q)+ \delta(A,q),
\ee
where $\delta(A,q)$ denotes finite correction terms. 
These are explicitly computed for several examples of HZ factorisable knots and links (see e.g. Table~\ref{tab:difference}). 

In Sec.~\ref{sec:BMW},   the origin of the correction terms $\delta$ is explained by the  Birman-Murakami-Wenzl (BMW) algebra representations, which give rise to the  Kauffman character expansion  \cite{Birman,Murakami1}.  
In the 3-strand case, for example, this reads
\be\label{SO(N)KF}
KF(\alpha,\beta) =\alpha^{-w} \bigg(\chi_0(\al,\beta)+\sum_{Q} \chi_Q(\beta) \eta_Q(\al,\beta)\bigg),
\ee
where the coefficients $\chi_Q$  coincide 
with the Racah coefficients $h^Q$ at $\beta=iq$, and $\eta_Q$ are the normalised quantum dimensions. 
The term $\chi_0$, which  gives rise to the correction $\delta$, is obtained as a trace of products of BMW algebra representations.  
The BMW algebra expansion allows to articulate explicit conditions for the HOMFLY--PT/Kauffman  relation (\ref{KFhatINTRO}) in terms of $\chi_Q$, which are explicitly 
given in the 3- and 4-strand cases in Sec.~\ref{sec:H/KF}. Such conditions are used to derive the above mentioned results regarding the 3-strand case of 
conjecture (\ref{conjINTRO}). 

Finally, in Appendix B, using the Jaeger's theorem \cite{Jaeger} we present a formal  expansion of  the state sum of  the
 Kauffman polynomial via $SU(N)$ characters.

\section{The Dubrovnik polynomial via characters of $SO(N+1)$}\label{sec:dubrovnik}

The  Kauffman polynomial  $KF(A,z)= A^{-w }L(A,z)$  can be combinatorially defined by the skein relation for $L(A,z)$
\ba\label{skeinKF}
&&L(\img{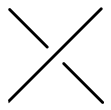})+L(\img{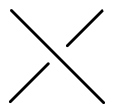})= z ( L(\img{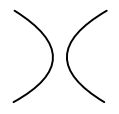})
+ L( \img{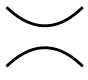}))\nonumber\\
&& L(\img{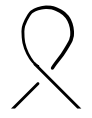}) = A L(\img{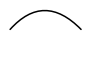}), \hskip 2mm L(\img{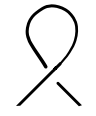})= A^{-1}L(\img{NonLoop.PNG});\;\;L(\bigcirc)=1.
\ea
The unnormalised version, labeled by $\overline{KF}$, is such that $\overline{KF}(\bigcirc)=\bar{L}(\bigcirc)=(\frac{A+A^{-1}}{z}-1)$. The Dubrovnik version of the Kauffman polynomial was defined as $\widetilde{KF}(A,z)=KF(iA,iz)$ in \cite{Petrou2}. 
There is an alternative version of the Dubrovnik polynomial, which is easier to express as a character expansion. This is
\be
Y(\mathcal{K};A,z)=A^{-w}\Lambda(\mathcal{K};A,z),
\ee
where  $\Lambda(\mathcal{K};A,z)$
is defined by a skein relation which differs slightly from  (\ref{skeinKF}), as
\be
\Lambda(\img{Kauffpositive.PNG})-\Lambda(\img{Kauffnegative.PNG})= z ( \Lambda(\img{Kauffzero.PNG})
- \Lambda( \img{KfInfinity.PNG}));\;\;\Lambda(\bigcirc)=1.
\ee
The unnormalised version is defined as $\bar{Y}(\bigcirc)=\bar{\Lambda}(\bigcirc)=(\frac{A-A^{-1}}{z}+1)$. 
The two versions 
can be related by\footnote{Note that for the normalised version $KF(\mathcal{K};A,z) = (-1)^{c+1}Y(\mathcal{K};iA,-iz)$ \cite{Kauffman}, that is, the normalisation factor can be thought of as having an extra component.}
 \be\label{KF-Y}
  \overline{KF}(\mathcal{K};A,z) = (-1)^{c}\bar{Y}(\mathcal{K};iA,-iz),
  \ee
 where $c$ is the number of components of the link $\mathcal{K}$.

As mentioned in the introduction, the character expansion for the  HOMFLY--PT polynomial  involving the Schur functions, which are the characters of  $SU(N)$, was outlined in \cite{Petrou3}. 
 In particular, according to the $SU(N)$ character expansion, the HOMFLY--PT polynomial $H(\mathcal{K})$ for a knot $\mathcal{K}$ with an $m$ strand braid representative 
 is expressed as \cite{Morozov,Petrou3}
\be\label{SU(N)}
H (\mathcal{K})=  A^{-w}\sum_{Q} h^{Q} (\mathcal{K})\frac{S_{Q}}{S_{[1]}}=:\mathbb{X}^{SU}(\mathcal{K}),
\ee
where the sum is over Young diagrams $Q$ with $m$ boxes, $h^{Q}$ are the so called the Racah coefficients, which are characters of the Artin braid group, and $S_{Q}$ are the Schur functions, that can be derived from the classical Weyl  formula. 

It is natural to attempt a similar expansion for the  Kauffman polynomial. 
Naively, one may approach this problem by leaving  the coefficients $h^Q$ unchanged, and simply replace $S_Q\to d_Q$, where $d_Q$ are the quantum dimensions  of $SO(N+1)$, to define
 \be\label{SO(N)}
{\mathbb{X}}^{SO}:= A^{-w}\sum_Q h^{Q}\frac{d_{Q}}{d_{[1]}}.
\ee
 A  formula to obtain the quantum dimensions $d_Q$ corresponding to a Young diagram $Q$, is given in Appendix A. 
For the fundamental representation, the Young diagram has a single box and it is denoted by $[1]$, for which $d_{[1]}$ reads 
\be\label{qd}
d_{[1]}
=1 + \frac{A-A^{-1}}{q-q^{-1}}=\bar{Y}(\bigcirc;A,q).
\ee
Note the difference between this expression and the $SU(N)$  case, for which $S_{[1]}= (A-A^{-1})/(q-q^{-1})=\bar{H}(\bigcirc;A,q)$.

The unnormalised Dubrovnik polynomial $\bar Y$ can be expressed in terms of $\bar {\mathbb{X}}^{SO}=d_{[1]}{\mathbb{X}}^{SO}$, 
plus  correction terms, denoted collectively by $\bar \delta$, as
\be\label{naivebarY}
\bar Y =\bar {\mathbb{X}}^{SO}+\bar \delta.
\ee
Notably, for  HZ factorisable knots (c.f. \cite{Petrou2}), 
 the term $\bar\delta$ is a simple polynomial, as can be seen through several examples that follow.
\paragraph{2 strands.}
For 2-strand braids the quantum dimensions $d_Q$, as given in Appendix A, are
\ba\label{d[2]}
&&d_{[2]}= (1 + \frac{\{q^2\}}{\{A q\}}) S_{[2]}; \hskip 2mm S_{[2]}=\frac{\{A\}\{Aq\}}{\{q\}\{q^2\}},\nonumber\\
&&\hspace{-6mm}d_{[11]}=(1+ \frac{\{q^2\}}{\{A q^{-1}\}}) S_{[11]}; \hskip 2mm S_{[11]}=\frac{\{A\}\{A q^{-1}\}}{\{q\}\{q^2\}}.
\ea
These 
satisfy the plethysm  (sum) rule $d_{[1]}^2= d_{[2]}+d_{[11]}+1$.  The Racah coefficients are $h^{[2]}=q^w$ and $h^{[11]}=(-q)^{-w}$, where $w$ is the writhe of the braid diagram.
The Dubrovnik polynomial $\bar Y(A,q)$ for 2-stranded  knots and links\footnote{Knots and links  correspond to odd and even $n$, respectively.}, which are always torus $T(2,n)$ with $w=n$,
can be expressed in terms of $d_{Q}$ by the character expansion in (\ref{naivebarY}) as  
\be\label{KauffmanCH}
 \bar Y(T(2,n);A,q) = A^{-n}\left(q^{n}d_{[2]}+ (-q)^{-n} d_{[11]}+ A^{-n}\right)
\ee
and hence the correction term is simply $\bar\delta=A^{-2n}$.

\vskip1mm
\noindent{}Remark. Since $d_Q$ in (\ref{d[2]}) can be divided into two parts corresponding to odd and even powers of $A$, the same holds for
 $ \bar Y (T(2,n)) $. When $n$ is odd, for instance, 
the odd part becomes $ \bar{Y}(T(2,n))_{\rm odd}= A^{-n}( q^{n} S_{[2]}-  q^{-n} S_{[11]})$, 
which coincides with the HOMPLY-PT character expansion 
in \cite{Petrou3}.

\vskip1mm
\noindent{}Example $T(2,3)$.
 The Dubrovnik polynomial, obtained by (\ref{d[2]}) and (\ref{KauffmanCH}), dividing with $d_{[1]}$, becomes
 \ba\label{K(2,3)}
 Y(T(2,3))
 \hspace{-1mm}&=& \hspace{-1mm} A^{-4} (1-q^2-\frac{1}{q^2})+A^{-2} (\frac{1}{q^2}+q^2)+ A^{-3} (q-\frac{1}{q})-A^{-5}(q-\frac{1}{q})\nonumber\\
  \hspace{-1mm}&=&  \hspace{-1mm}A^{-4}(-1-z^2)+A^{-2} (2+z^2)+A^{-3}z-A^{-5} z,
 \ea
 which, by recalling that
 $Y(iA,-iz)=KF(A,z)$,
  coincides with the standard Kauffman polynomial 
  $KF(3_1)=(A^{-5}+ A^{-3})z + A^{-4}(z^2-1)+A^{-2}(z^2-2)$ (see e.g. \cite{Bar-Natan}). 

\paragraph{3 strands.}
The  Young diagrams with $3$ boxes are $[3],[21],[111]$.
The corresponding quantum dimensions  for $SU(N)$ are
\ba\label{schur3}
&&S_{[3]}= \frac{\{A\}\{A q\}\{A q^2\}}{\{q\}\{q^2\}\{q^3\}},\;\; S_{[21]}= \frac{\{A\}\{A q\}\{A q^{-1}\}}{\{q^2\}^2\{q^3\}}\nonumber\\
&&\hspace{12mm} S_{[111]}= \frac{\{A\}\{A q^{-1}\}\{A q^{-2}\}}{\{q\}\{q^2\}\{q^3\}},
\ea
while for $SO(N+1)$ they become
\ba\label{3strand1}
&&d_{[3]} 
=( 1+ \frac{\{q^3\}}{\{A q^2\}}) S_3,\;\;d_{[21]} 
= (1+ \frac{\{q^3\}}{\{A\}})S_{21}\nonumber\\
&&\hspace{12mm}d_{[111]} 
=(1+ \frac{\{q^3\}}{\{A q^{-2}\}})S_{111}.
\ea
The plethysm  rule in this case becomes $d_{[1]}^3= d_{[3]}+ 2 d_{[21]}+ d_{[111]}+3 d_{[1]}$.

For torus knots $T(3,n)$ (i.e. $n\neq3k$) with braid  $(\sigma_2\sigma_1)^n$, the HOMFLY--PT character expansion reads
\be
\bar{H}(T(3,n))= A^{-2n}(q^{2n}S_{[3]} -S_{[21]} + q^{-2n}S_{[111]}).
\ee
Here the Racah coefficients $h^{[3]}=q^w$ and $h^{[111]}=(-q)^{-w}$ are the same as before, while  $h^{[21]}=-1$  
was computed in \cite{Petrou3} as the trace or the $R$-matrices
$(R_2R_1)^n$.
The Dubrovnik polynomial is expressed by the character expansion
\be\label{K(3,n)}
\bar{Y}(T(3,n);A,q) = A^{-2n} \left(q^{2n} d_{[3]}- d_{[21]}+ q^{-2n} d_{[111]}\right).
\ee
Remarkably, this  is exactly the same as ${\mathbb{X}}^{SO}$ and hence $\bar\delta=0$.
This result can also be derived from the Chern-Simons theory \cite{Stevan}.  
For torus links $T(3,3k)$, with $k\geq0$, we find\footnote{This result correctly yields the relation $H(T(3,3k))-\widehat{KF}(T(3,3k))=-3A^{8k}$, found in \cite{Petrou2}.} 
\be\label{KT(3,n)CH}
{\bar{Y}(T(3,3k);A,q)= A^{-6k}(q^{6k}d_{[3]}+2d_{[21]}+q^{-6k}d_{[111]})+3A^{-8k}d_{[1]}.}
\ee

For hyperbolic knots which are HZ-factorisable, the term $\bar \delta(A,q)$ remains relatively simple as can be seen, e.g., in
 Table~\ref{tab:difference}. 

   \begin{table}[htbp!]
\caption{The correction terms $\bar\delta=\bar Y-\bar{\mathbb{X}}^{SO}$ for all  3-strand HZ-factorisable knots $\mathcal{K}$ with up to 13 crossings.}\label{tab:difference} 
\hfill\begin{tabular}{ll}
\toprule
\scalebox{0.92}[0.92]{$\mathcal{K}$} & \scalebox{0.8}[0.8]{$\bar \delta= \bar Y-  \bar {\mathbb{X}}^{SO}$} \\
\midrule
$5_2$ &  \scalebox{0.91}[0.91]{ $A^2+A^8(q^{-2}+ q^2)- A^6(q^{-4}+1+ q^4)$} \\
\hline
$8_{20}$& \scalebox{0.91}[0.91]{  $A^{-2} + A^6(q^{-4} + 1 + q^4)
- A^4 (q^{-6}+ q^{-2} +q^2 + q^6)$}
\\
\hline
$10_{125}$& \scalebox{0.91}[0.91]{ $A^{-6} + A^4 (q^{-6} + q^{-2} + q^2 + q^6) - 
 A^2 (q^{-8} + q^{-4} + 1 + q^{4} + q^{8})$}\\
\hline
$10_{139}$ & \scalebox{0.91}[0.91]{ $A^{-16} (q^{-2}+q^2) + A^{-10} - A^{-14}(q^{-4} + 1 + q^4)$}
\\
\hline
$10_{161}$ & \scalebox{0.91}[0.91]{ $A^6 +A^{14}(q^{-4}+1 + q^4)-A^{12}(q^{-6}+q^{-2}+q^2 +q^6)$ }\\
\hline
$12n_{235}$ & \scalebox{0.91}[0.91]{$
A^{-2}(q^{-8}+ q^{-4} + 1 + q^4 + q^8) + A^{10} - (q^{-10}+ q^{-6} + q^{-2} + q^2 + q^6 +q^{10})$} \\
\hline
$12n_{242}$ & \scalebox{0.91}[0.91]{ $ A^{-18} + A^{-14} - A^{-16}(q^{-2}+q^2)$}\\
\hline
$12n_{749}$ & \scalebox{0.91}[0.91]{   $A^{-2}+ A^{-12} (q^{-6} + q^{-2} + q^2 + q^6) - A^{-10}(q^{-8}+ q^{-4} + 1 + q^4 + q^8)$}
 \\
\bottomrule
\end{tabular}
   \end{table}
   
  \noindent{}From this table, we observe the following relations for knots related to each other by full twists (i.e. $5_2\otimes F_3=10_{139}$, $8_{20}\otimes F_3=10_{161}$ and  $10_{125}\otimes F_3=12n_{749}$; c.f. Fig.~2 in \cite{Petrou2}): $A^{-8}\bar\delta(10_{125};A^{-1},q)=\bar\delta(12n_{749};A,q)$, $A^{-8}\bar\delta(10_{161};A,q)=  \bar\delta(8_{20};A,q)$ and $A^{-8}\bar\delta(10_{139};A^{-1},q)=\bar\delta(5_2;A,q)$.

\vskip2mm
\noindent{$\circ$ Example $5_2^{3k}:=5_2\otimes F_3^k$.}
This family is obtained by attaching $k$ full twists $F_3^k:=(\sigma_2\sigma_1)^{3k}$ to a 3-strand braid representative of $5_2$, such as $\sigma_1\sigma_2^3\sigma_1\sigma_2^{-1}$  \cite{Petrou2}.  Its Dubrovnik character expansion 
for $k\geq1$, becomes
\ba\label{char52F3k}
\bar{Y}(5_2^{3k};A,q)&=& A^{-4-6k}(q^{4+6k}d_{[3]}+h^{[21]}d_{[21]}+q^{-4-6k}d_{[111]})\\\nonumber
&&+A^{-2-8k}-A^{-6-8k}(q^4+1+q^{-4})+A^{-8(k+1)}(q^2+q^{-2}),
\ea
where $h^{[21]}=-q^4+q^2-1+q^{-2}-q^{-4}$ for all members of the family, since this Racah coefficient remains unchanged under full twists \cite{Petrou3}.

\vskip2mm
\noindent{}$\circ$ Example $\mathcal{K}_{j,2}:=\sigma_1\sigma_2^{-2j-1}\otimes E_3^2$. With braid $\sigma_1\sigma_2^{-2j-1}(\sigma_1\sigma_2^2\sigma_1)^2$, $w=8-2j$ and $h^{[21]}=\sum_{i=0}^{2j+6}(-1)^{i+1}q^{2(j-i)+6}$ for $j\geq0$, we find
\be\label{charP23j}
\bar{\delta}(\mathcal{K}_{j,2};A,q)=
A^{4j-6}-A^{2j-12}\left(\sum_{i=0}^{j+3}q^{2j+6-4i}\right)+A^{2j-14}\left(\sum_{i=0}^{j+2}q^{2j+4-4i}\right).
\ee

It is noteworthy that in the Table~\ref{tab:difference} and in the formulas (\ref{char52F3k})-(\ref{charP23j}) the $q$-polynomials that appear in the parentheses in the correction term $\bar\delta(A,q)$, contain the negative and positive terms appearing in $h^{[21]}$. As a result, at $A=1$,  $\bar\delta(1,q) = h^{[21]}+1.$ Moreover, 
 $\bar\delta=0$ when 
 $A=q$. These properties do no longer hold in the case of links, as can be seen in the examples below.

\vskip2mm
\noindent{}$\circ$  Example $L7n1\{0\}$.  $\sigma_1^{-1}\sigma_2^{-1}\sigma_1^{-1}\sigma_2^{-1}\sigma_1^{-1}\sigma_2^{-2}$, $ w=-7$, $h^{[21]}=q^{-1} - q$
\be\label{L7n1}
\bar{\delta}(L7n1\{0\})=
A^8 + \frac{A^{11}}{ q- q^{-1}} -\frac{A^9 ( q^{-2}-1 + q^2)}{ q- q^{-1}}.
\ee
\vskip2mm
\noindent{}$\circ$  Example $L7n2\{0\}$. $\sigma_1^{-1}\sigma_2\sigma_1^{-1}\sigma_2^{-1}\sigma_1\sigma_2^{-2}	$, $ w=-3$, $h^{[21]}=q^{-5} - q^{-3} + q^{-1} - q + q^3 - q^5$
\ba\label{L7n2}
\bar{\delta}(L7n2\{0\})&=&
+1+
\frac{A^7 (q^{-4} - q^{-2} + 1 - q^2 + q^4)}{q-q^{-1}}\\
&&+\frac{ A^5 (-q^{-6} + q^{-4} - q^{-2} + 1 - q^2 + q^4 - q^6)}{q-q^{-1}} 
.\nonumber
\ea
  
  \noindent{}Remark. For all 3-strand knots and links, ${\mathbb{X}}^{SO}$ 
  and  ${\mathbb{X}}^{SU}$ 
  are related  by
  \ba
 {\mathbb{X}}^{SO}-{\mathbb{X}}^{SU}\hspace{-2mm} &=& \hspace{-2mm}  A^{-w}\sum_Q h^{Q}\left(\frac{d_{Q}}{d_{[1]}}- \frac{S_{Q}}{S_{[1]}}\right) \nonumber\\
  &=&\hspace{-2mm} A^{-w}\left(\frac{(A-q)(1+A q)}{A(q^2-1)(1+q^2 + q^4)}(q^{2+w}+ h^{[21]}(1+q^4) + q^{2-w})\right)\nonumber\\
 &=& \hspace{-2mm}    \frac{A^{-w}}{q^2+1+q^{-2}}\left( \frac{\{A\}}{\{q\}}-1\right) ( q^w + h^{[21]}(q^2+q^{-2}) + q^{-w}).
 \ea

\paragraph{4 strands.}
There are five Young diagrams with 4 boxes. The corresponding quantum dimensions
$d_{[4]},d_{[31]},d_{[22]},d_{[211]}$ and $d_{[1111]}$ are given in Table~\ref{tab:d_Q}  of Appendix A. 
In general $h^{[4]}=q^w$ $h^{[1111]}=(-q)^{-w}$ and hence,
\be
{\mathbb{X}}^{SO}=A^{-w}\left(q^{w}d_{[4]}+h^{[31]}d_{[31]}+h^{[22]}d_{[22]}+h^{[211]}d_{[211]}+(-q)^{-w}d_{[1111]}\right),
\ee
where  $h^{[211]}$ can be obtained from $h^{[31]}$ by $q\to -q^{-1}$.

The 4-strand torus knots $T(4,n)$, for $n$ odd, have character expansion
\be\label{T(4,n)}
\bar Y(T(4,n))=A^{-3 n} (q^{3 n} d_{[4]} - q^{n} d_{[31]} + q^{-n} d_{[211]} - q^{-3 n} d_{[1111]})+A^{-4n}
\ee
where the Racah coefficients $h^{[31]}=-q^{n}$ and $h^{[22]}=0$ are obtained in \cite{Petrou3}. 
The simple correction term $\bar\delta=A^{-4n}$ is similar to the 2-strand case in (\ref{KauffmanCH}).  
For 4-strand torus links, however, corresponding to even $n$, for which the HZ transform is not factorisable, the correction term $\bar\delta$ becomes much more complicated. 

Below we present the correction terms $\bar\delta$ for several examples of 4-strand hyperbolic  knots that admit HZ-factorisability. 
\vskip2mm
\noindent{}$\circ$ Example $10_{132}$. $\sigma_1^{-3}\sigma_2\sigma_1^2\sigma_2\sigma_3\sigma_2^{-1}\sigma_3^2$; $w=3$, $h^{[31]}=-q^{-7}+q^{-5}-q^{-1}+q-q^5$ 
\ba
\bar{\delta}(10_{132})\hspace{-2mm}&=& \hspace{-2mm}
 A^{-8}(q^{-6} + q^{-2} + 1 + q^2 + q^6)
  \nonumber\\\nonumber
&&\hspace{-13mm} - A^{-6}(q^{-8} + q^{-6} + q^{-4} +q^{-2} +2  +q^2 +q^4 +q^6 + q^8)\\
&&\hspace{-13mm}+A^{-4} (q^{-8} + q^{-4} + q^{-2} + 1 + q^2+ q^4 + q^8)- A^{-2}(q^{-2} + q^2) +1.
\ea
\vskip1mm
\noindent{}$\circ$ Example  $10_{132}\otimes F_4$. $w=15$,  $h^{[31]}=-q^{-3} + q^{-1} - q^3 + q^5 - q^9$
\ba\label{10_132F4}
\bar{\delta}(10_{132}\otimes F_4)\hspace{-2mm}&=& \hspace{-2mm}
A^{-14}(q^2+q^{-2})+A^{-18}(q^6+1+q^{-6})+A^{-22}(q^4+1+q^{-4})\nonumber\\
&&\hspace{-2mm}-A^{-16}(q^4+1+q^{-4})-A^{-20}(q^6+q^2+q^{-2}+q^{-6}).
\ea
\vskip1mm
\noindent{}$\circ$  Example $10_{128}$.  $\sigma_1^3\sigma_2\sigma_1^2\sigma_2^2\sigma_3\sigma_2^{-1}\sigma_3$; $w=9$, $h^{[31]}=-q^{-1} + q - q^3 + q^7 - q^9$
\ba\label{10_128}
\bar{\delta}(10_{128})\hspace{-2mm}&=& \hspace{-2mm}
A^{-6}( q^{-4}+1 + q^4) + A^{-8}( - q^{-6} - 2q^{-2}+1 - 2 q^2 - q^6) 
  \nonumber\\
&&\hspace{-10mm}+A^{-10} (  q^{-2}-1 + q^2) + A^{-12} ( q^{-6} - q^{-4} + q^{-2} + q^2 - q^4 + q^6).
\ea
\vskip1mm
\noindent{}$\circ$ Example  $T(4,5,3,3)$.
\be
\bar Y(T(4,5,3,3))=A^{-21} (q^{21} d_{[4]} -q^7 d_{[31]} + q^{-7} d_{[211]} - q^{-21} d_{[1111]})+A^{-28}.
\ee

\noindent{Remark.} For the exceptional 2-component link $L10n_{42}\{1\}$, $\bar\delta$ is  no longer simple (involves 12  different powers of $A$).  This case is exceptional, as explained in \cite{Petrou2}, because its HOMFLY--PT satisfies HZ-factorisability but it is not related to the Kauffman polynomial by the link analogue of (\ref{KFhatINTRO}), as explained in more detail 
in Sec.~\ref{sec:H/KF}. This implies, that it is  the validity of a  HOFMLY--PT/Kauffman, which for two-compoent links is stronger than HZ factorisability,  that determines whether  the correction terms $\bar\delta$ are simple as in the above examples, or not. 

\paragraph{5 strands.} The quantum dimensions for 5-strand are given in Appendix A.
For torus knots the expansion is simply given by ${\mathbb{X}}^{SO}$ (i.e. $\bar\delta=0$), that is
\be\label{boxT(5,n)}
\bar Y(T(5,n))= A^{-4n}(q^{4n}d_{[5]}-q^{2n}d_{[41]}+d_{[311]}- q^{-2n}d_{[2111]}+ q^{-4n}d_{[11111]}).
\ee

\noindent{Remark.}
In general, for \emph{knots}  holds
\be\label{Kchar}
{{\mathbb{X}}^{SU} = {\mathbb{X}}^{SO}_{\rm even}+ \frac{q-q^{-1}}{A-A^{-1}}{\mathbb{X}}^{SO}_{\rm odd}}
\ee
and hence, as anticipated in the previous remark, the HOFMLY--PT/Kauffman relation $H=\widehat{KF}$  implies a condition on the difference terms $\delta$. Hence, in order to study  this relation in more depth, it is necessary to explain the representation theoretic origin of the difference terms, to which we turn in the next section.

\section{BMW algebra representations and the Kauffman character expansion }\label{sec:BMW}
It has been long established \cite{Birman,Murakami1} that the   Kauffman polynomial $KF(\alpha,\beta)$ admits a character expansion, in which the coefficients are determined by representations of the Birman-Murakami-Wenzl  (BMW) algebra $C_m(\al,\beta)$. 
This algebra is generated by the braid group generators $\sigma_i\in B_m$ and by an extra set of generators $e_i$, which is depicted in Fig.~\ref{ei}.
\begin{figure}[h!]
     \centering
\includegraphics[scale=0.2]{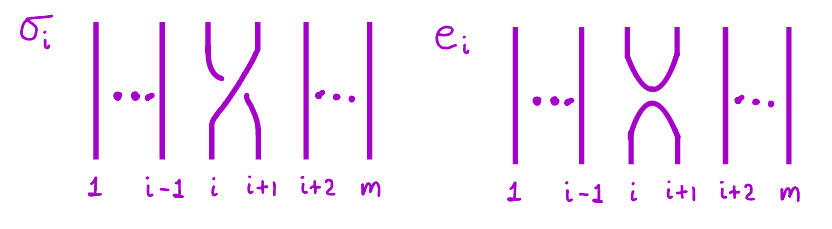}
     \caption{ The BMW algebra generators $\sigma_i$ and $e_i$.} 
     \label{ei}
 \end{figure}
Combinations of the $\sigma_i$ and $e_i$ generators form a \emph{knit} \cite{Murakami1}, which is a generalisation of a braid. 
These generators satisfy the following relations \\(i) the braid relations 
\ba
&&\sigma_i\sigma_{i+1}\sigma_{i}=\sigma_{i+1}\sigma_i \sigma_{i+1}\nonumber\\
&&\sigma_i\sigma_j=\sigma_j\sigma_i ;\;\;|i-j|\geq 2
\ea
  (ii) the Kauffman  skein relation
\be\label{BMWskein}
\sigma_i+\sigma_i^{-1}= (\beta+\beta^{-1})(1+ e_i)\ee 
(iii) the idempotent relation\footnote{The relation (\ref{idempotent}) is reminiscent of a projection operator, which satisfies $P^2=P$, and also appears when there is gauge invariance for a Cartan symmetric (quotient) space 
\cite{ZinnJustin}.} 
\be\label{idempotent}
e_i^2=\mu e_i,\;\;\mu:=\frac{\alpha+\alpha^{-1}}{\beta+\beta^{-1}}-1.\ee
If (\ref{BMWskein}) is considered as the definition of $e_i$, the BMW algebra $C_m(\al,\beta)$ can be  defined as the quotient of $\mathbf{C}[B_m]$ by the relations  \cite{Wenzl}
\be 
e_i\sigma_i = \alpha^{-1}e_i, \;e_i\sigma^{\pm1}_{i-1}e_i = \alpha^{\pm1}e_i\;\forall i.
\ee

The Kauffman polynomial  $KF(x;\alpha,\beta)= \alpha^{-w}L(x;\alpha,\beta)$ of an $m$-strand knit $x$ is expressed 
as {\cite{Birman,Murakami1} (see also Thm.~5.1.5 in \cite{Murakami3}, or Thm.~2.2. in \cite{Murakami2}) 
\be\label{KF-BMW}
KF(x;\alpha,\beta)=(-1)^{m-1}\alpha^{-w}\sum_{j=0}^{\lfloor m/2\rfloor}\sum_{Q,|Q|=m-2j}\chi_{Q}(x;\alpha,\beta)\eta_{Q}(\alpha,\beta)
\ee
where $|Q|$ is the number of boxes of the Young diagram $Q$,
$\chi_{Q}$ are the characters of irreducible representations $\rho_Q$ of the BMW algebra, and\footnote{In the literature, the extra factor $(-1)^{m-1}$ in (\ref{KF-BMW}) is  often absorbed in the definition of $\eta_Q(\al,\beta)$.}
$\eta_{Q}(\al,\beta)=\frac{d_{Q}(i\alpha,-i\beta)}{d_{[1]}(i\alpha,-i\beta)}$ are the normalised quantum dimensions of $SO(N+1)$. In particular, the Kauffman character expansion for a knit with $m$ strands, in contrast with the $SU(N)$ case,  does not involve only Young diagrams with $m$ boxes, but also with $m-2,m-4,...,\epsilon$ boxes, 
where $\epsilon=1$ or $0$, depending on if $m$ is odd or even, respectively. 

\begin{figure}[h!]
    \centering
\includegraphics[width=\linewidth]{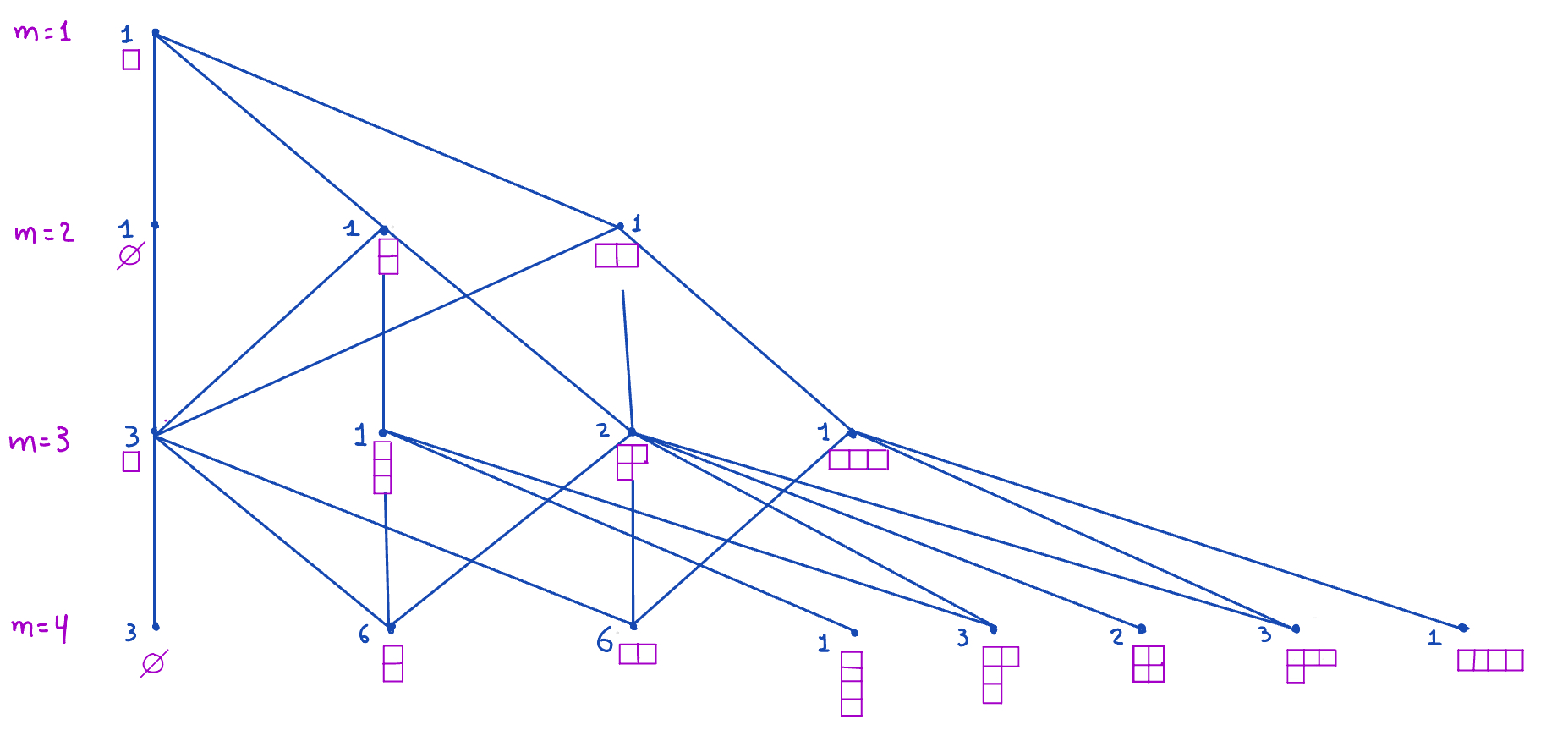}
    \caption{The Bratelli diagram, with the dimensions of the representations of $C_m(\al,\beta)$ indicated next to each Young diagram.}
    \label{fig:Bratelli}
\end{figure}

For fixed number of strands $m$, the dimension of the representation $\rho_Q$ corresponding to a Young diagram $Q$, is given by the number of paths in  the Bratelli diagram, which lead from $Q$ to the diagram $\square$, at $m=1$  \cite{Turaev}. The Bratelli diagram for up to $m=4$  is  depicted in Fig.~\ref{fig:Bratelli}, in which the dimension of each representation that contributes in the $m$-strand expansion is indicated. 
 Remarkably, these dimensions are also reflected in the plethysm formulas 
$d_{[1]}^2= d_{[2]}+d_{[11]}+ d_{[1]}, $ $d_{[1]}^3=d_{[3]}+ 2 d_{[21]}+d_{[111]}+ 3 d_{[1]},$ $
d_{[1]}^4= d_{[4]}+ 3 d_{[31]}+ 2 d_{[22]}+ 3 d_{[211]}+ d_{[1111]}+ 6 d_{[2]}+ 6 d_{[11]}+3$ \cite{Littlewood}.

\paragraph{2 strands.} The character expansion of the Kauffman polynomial  involves a sum over Young diagrams with 2 boxes and with 0 boxes, as indicated in the Bratelli diagram. For a knit $x$ 
it is expressed  by 
\be\label{KF2strand}
{KF(x;\al,\beta)=-\alpha^{-w}(\chi_0(x)\eta_0+\chi_{[2]}(x)\eta_{[2]}+\chi_{[11]}(x)\eta_{[11]}).}
\ee
The normalised quantum dimensions  $\eta_Q(\al,\beta)=\frac{d_{Q}(i\alpha,-i\beta)}{d_{[1]}(i\alpha,-i\beta)}$
are given explicitly by 
\ba\label{KFdim}
\eta_0(\alpha,\beta)&=&\frac{1}{d_{[1]}(i\alpha,-i\beta)}=\frac{\alpha (1+ \beta^2)}{( \beta\alpha-1)( \beta-\alpha)},\nonumber\\
\eta_{[2]}(\alpha,\beta)&=& \frac{(\alpha+\alpha^{-1}) \beta (\alpha \beta^3+1)}{(1-\beta^2)(\alpha \beta-1)( \beta^2+1)},\nonumber\\
\eta_{[11]}(\alpha,\beta)&=&\frac{(\alpha+\alpha^{-1})\beta (\alpha+ \beta^3)}{(\beta^2-1) (\alpha- \beta) (\beta^2+1)}.
\ea
The BMW algebra has 2 generators $\{\sigma_1,e_1\}$ and their 1-dimensional representations, as suggested in the Bratelli diagram in Fig.~\ref{fig:Bratelli}, are 
\be
\rho_0(\sigma_1)=\alpha^{-1}, \;\;\rho_0(e_1)=\mu,\;\;\rho_{[2]}(\sigma_1)=\beta,\;\;\rho_{[11]}(\sigma_1)=\beta^{-1}, \;\;\rho_{1,2}(e_1)=0.
\ee
These are such that to satisfy the defining relations and determine the coefficients for a knit $x$ 
by $\chi_Q(x)=tr(\rho_Q(x))$.

\vskip2mm
As mention in the previous section, the only possible 2-strand knots and links have braid of the form $\sigma_1^n$ and belong to the torus family $T(2,n)$. For these we compute $\chi_0=tr(\rho_0(\sigma_1)^n)=\al^{-n}$, $\chi_{[2]}=tr(\rho_{[2]}(\sigma_1)^n)=\beta^{n}$ and $\chi_{[11]}=tr(\rho_{[11]}(\sigma_1)^n)=\beta^{-n}$ and hence, the 2-strand expansion can be in general expressed as
\be\label{chi(3_1)}
KF(\sigma_1^n;\alpha,\beta)=- \alpha^{-n}(\beta^n \eta_{[2]}+\beta^{-n}\eta_{[11]}+ \alpha^{-n}\eta_0).
\ee
If we set $\alpha=-iA$ and $\beta=iq$, and using that $Y(A,q)=(-1)^{c+1}KF(-iA,iq)$, where $c$ is the number of components, this gives a result consistent with (\ref{KauffmanCH}).

\paragraph{3 strands.}
The Kauffman polynomial  of a 3-strand knit $x$ involves Young diagrams with 3 and 1 boxes, hence it is expressed 
as {\cite{Birman,Murakami1}
\be\label{new}
{KF(x;\alpha,\beta)=\alpha^{-w}(\chi_0(x)+ \chi_{[3]}(x)\eta_{[3]}+\chi_{[21]}(x)\eta_{[21]}+ \chi_{[111]}(x)\eta_{[111]}).}
\ee
Here we used that $\eta_{[1]}=1$ and label $\chi_0:=\chi_{[1]}$, while for $|Q|=3$ the 
 $\eta_Q$ (denoted by 
$a_{1,2,3}$ in \cite{Murakami1}) 
can be written in a factorised form as
\ba\label{ai}
&& \eta_{[3]}= \frac{(\alpha^2+1)\beta^2 (1+ \alpha \beta)(\alpha \beta^5-1)}{\alpha^2  (1+ \beta^2)^2(\beta^2-1)(1-\beta^2+\beta^4)}\nonumber\\
&&\eta_{[21]}= \frac{(1-\beta^2)(\alpha+\beta)(1+\alpha \beta)(\alpha+\beta^3)(1+\alpha \beta^3)}{\alpha^2 \beta^3(1+\beta^2)^2(\beta-\beta^{-1})(\beta^2-1+\beta^{-2})}\nonumber\\
&&
\eta_{[111]}= \frac{(1+ \alpha^2)\beta^2(\alpha+ \beta)(-\alpha+\beta^5)}{\alpha^2  (1+\beta^2)^2(\beta^2-1)(1-\beta^2+\beta^{4})}.
\ea
Note that under $\beta\to\beta^{-1}$,   $\eta_{[3]}$ is symmetric to $\eta_{[111]}$ 
and $\eta_{[21]}$ to itself, 
 in agreement with the reflection symmetries of the respective Young diagrams.
 
The coefficient
$\chi_0$  is obtained as the trace of the 3-dimensional (c.f. Fig.~\ref{fig:Bratelli}) irreducible representations $\rho_0$  
given by \cite{Murakami3}\footnote{\label{foot:r0s2}The $\rho_0(\sigma_i)$  in (\ref{r0rep}) match with the ones in Prop.~4.4.8 of \cite{Murakami3}, with $x=z$. An alternative expression given in \cite{Murakami1} is $\rho_0(\sigma_1)= \begin{pmatrix} 
  \alpha^{-1}&z&0\\
  0&z&1\\
 0&-1&0
 \end{pmatrix}$ and $\rho_0(\sigma_2)= \begin{pmatrix} 
  0&0&-1\\
  0&\alpha^{-1}&\alpha^{-1} z\\
  1&0&z
\end{pmatrix}$, where we have corrected   $\rho_0(\sigma_2)$ 
 to satisfy the defining relation (ii). }  
\ba\label{r0rep}
&&\rho_0(\sigma_1)= \begin{pmatrix} 
  \alpha^{-1}&0&z\\
  0&0&-1\\
  0&1&z
  \end{pmatrix},\hskip 3mm
  \rho_0(\sigma_2)= \begin{pmatrix} 
  0&-1&0\\
  1&z&0\\
  0&\alpha^{-1}z&\alpha^{-1}
\end{pmatrix}\nonumber\\
&&\hspace{7mm}  \rho_0(e_1)= \begin{pmatrix} 
  \mu&\alpha&1\\
  0&0&0\\
  0&0&0
  \end{pmatrix}, \hskip 3mm
  \rho_0(e_2)= \begin{pmatrix} 
  0&0&0\\
  0&0&0\\
1&\al^{-1}&\mu
  \end{pmatrix},
\ea
where $z=\beta+\beta^{-1}$ and $\mu=\frac{\alpha+\alpha^{-1}}{z}-1$.
The remaining coefficients $\chi_Q$ are determined similarly by the representations \cite{Murakami1,Birman} 
\ba\label{matrixrepresentation}
  &&\rho_{[3]}(\sigma_1)= \rho_{[3]}(\sigma_2)=\beta, \hskip 2mm
  \rho_{[3]}(e_1)=
  \rho_{[3]}(e_2)=0\nonumber\\
  &&\rho_{[21]}(\sigma_1)=\begin{pmatrix} 
  \beta&i\\
  0&\beta^{-1}
  \end{pmatrix},\;\;\rho_{[21]}(\sigma_2)=\begin{pmatrix}
  \beta^{-1}&0\\
  i&\beta
  \end{pmatrix},\nonumber\\
  &&\hspace{15mm}\rho_{[21]}(e_1)=\rho_{[21]}(e_2) = \begin{pmatrix} 
  0&0\\
  0&0
  \end{pmatrix}\nonumber\\
  &&\rho_{[111]}(\sigma_1)= \rho_{[111]}(\sigma_2)=\beta^{-1},\hskip 2mm \rho_{[111]}(e_1)=\rho_{[111]}(e_2)=0.
\ea
In particular, for a knit $x=\sigma_1^{r_1}\sigma_2^{s_1}\cdots$, with $r_1,s_1\in\mathbb{Z}$, the coefficients are 
\be\label{coeff3}
\chi_{Q}(x)= \text{tr} (\prod \rho_{Q}(\sigma_1)^{r_1}\rho_{Q}(\sigma_2)^{s_1}\cdots).
\ee 
The  matrices in (\ref{r0rep}) and (\ref{matrixrepresentation})  
satisfy the defining relations (i)-(iii) of the BMW algebra and the further relations, for $i=1,2$,
\ba
&&\text{det}\rho_0(\sigma_{i})= \alpha^{-1}, \hskip 2mm
 \text{tr}\rho_0(\sigma_i)= \alpha^{-1}+z\nonumber\\
&&    K^{(2)} \rho_2(\sigma_1) K^{(2)}= \rho_2(\sigma_2);\;\;K^{(2)}:= \begin{pmatrix}
    0&1\\
    1&0
    \end{pmatrix}\nonumber\\
&&\text{tr}(e_i^n)= \mu^n, \hskip 3mm\text{tr}(e_i \sigma_i)= \alpha^{-1} \mu \hskip 3mm
\text{tr}(e_i \sigma_i^{-1})=\alpha \mu.
    \ea
  
\noindent{Remark.} For $|Q|=m=3$, the coefficients (\ref{coeff3}) depend only on $\beta$, as can be seen in (\ref{matrixrepresentation}), and
 coincide (up to factors $\pm i$) with the  Racah coefficients $h^{Q}$ from the $SU(N)$ expansion at $\beta=iq$, as anticipated in the naive approach of the previous section. Hence, the Racah coefficients can be determined either  by the  real matrices $R_i$ in \cite{Petrou3},  or through the complex matrices $\rho_Q(\sigma_i)$ in (\ref{matrixrepresentation}).  For $|Q|<m$, however, $\chi_{Q}$ depend both on $\alpha$ and $\beta$, as can be seen in (\ref{r0rep}). 
 \begin{prop}\label{prop:delta-x0} The 3-strand correction terms $\delta=\bar\delta/d_{[1]}$ in Sec.~\ref{sec:dubrovnik}
can be obtained in terms of BMW algebra representations by 
 \be\label{boxeddelta}
 {\delta(A,q) = (iA)^{-w}\chi_0(x;-iA,iq)}.
 \ee
  \end{prop}
  \begin{proof}
     By direct comparison of the Dubrovnik character expansion $Y(A,q)= {\mathbb{X}}^{SO}(A,q)+ \delta(A,q)$ with (\ref{new}), using the relation $Y(A,q)=(-1)^{w}KF(-iA,iq)$ from (\ref{KF-Y}), since the parity of $c$ coincides with $w+1$, and that $\chi_Q(iq)=i^wh^Q(q)$. 
  \end{proof}
 \begin{cor}
The character expansion of the unnormalised Dubrovnik polynomial for a 3 strand knit can be expressed as
 \be\label{Thm5.1}
 \bar{Y}(A,q)=A^{-w}\left(\sum_Qh^Q(q)d_Q(A,q)\pm d_{[1]}(A,q)\chi_0(-iA,iq)\right).
 \ee
 \end{cor}

\vskip2mm
\noindent{$\circ$} Example $3_1$. For a 3-strand braid $ x=(\sigma_1\sigma_2)^2$  with $w=4$  we find
\be
\chi_{[3]}(x)= \beta^4, \hskip 2mm
\chi_{[21]}(x) = -1,\hskip 2mm
\chi_{[111]}(x)=\beta^{-4}
\ee
which, with $\beta=iq$, are the same as the $SU(N)$
Racah coefficients $h^{Q}(3_1)$.
The coefficient  $\chi_0(x)=\text{tr}\left((\rho_0(\sigma_1)
\rho_0(\sigma_2))^2\right)=0$, is  vanishing. 
The Kauffman polynomial $KF(\alpha,\beta)$,  with $z=\beta+\beta^{-1}$, hence becomes 
\ba\label{KF31}
KF(\alpha,\beta)&=&\al^{-4}(\eta_{[3]} \beta^4- \eta_{[21]}+ \beta^{-4} \eta_{[111]})\nonumber\\
&=&\alpha^{-4}(1-z^2)+ \alpha^{-2}(2-z^2)
-(\alpha^{-5}+\alpha^{-3})z,
\ea
in agreement with (\ref{K(2,3)}), after $\al=-iA,\beta=iq$. 
Alternatively, $3_1$ can be thought of as the closure of the knit
$e_1 \sigma_2^3 \sigma_1$, which has writhe\footnote{Note that for knits, the writhe is not the sum of the exponents of $\sigma_i$ and $e_i$. This can be understood by the fact that $e_i$ reverses the orientation of a strand. To determine the writhe of a knit, it is hence necessary to draw the oriented braid diagram and count the signed crossings.} $w=2$ and for which
\be
\chi_0(3_1)=\text{tr}(e_1 \sigma_2^3 \sigma_1)= -2+ z^2 + \alpha^{-2}(z^2-1)+\alpha^{-1} z + \alpha^{-3} z
\ee
This yields the Kauffman polynomial as $KF(\al,\beta)=\al^{-w}\chi_0$ \cite{Murakami1}.

\vskip 2mm
\noindent{$\circ$} Example  $T(3,n)$. 
The Kauffman polynomial  for 3-strand torus knots is determined by (\ref{new}) with $\chi_0(x)=\text{tr}(\rho_0(\sigma_1)\rho_0(\sigma_2))^{n}=0.$ By Prop.~\ref{prop:delta-x0}, this is equivalent to the statement that the Kauffman polynomial of $T(3,n)$ is given by ${\mathbb{X}}^{SO}$ without correction terms ($\delta=0$).

\vskip 2mm
\noindent{\bf $\circ$}
Example $5_2$.
With braid $x=\sigma_1^{-3}\sigma_2^{-1}\sigma_1\sigma_2^{-1}$, $w=-4$, we compute $\chi_{[3]}(x)= \beta^{-4}$
, 
 $\chi_{[21]}(x)= -(\beta^4+\beta^2+ 1+ \beta^{-2}+\beta^{-4})$  (in agreement with $h^{[21]})$ 
and $\chi_0=  (\beta^{-3} + \beta^{-1} + \beta + \beta^3)\al^{3} + ( \beta^{-4} + 2\beta^{-2} +2+ 2 \beta^2 + \beta^4)\al^{2}+ (
 \beta^{-3} + 2\beta^{-1} + 2 \beta + \beta^3)\al +\beta^{-2} +2+ \beta^2 +  (\beta^{-1} + \beta)\al^{-1} $. 
 In agreement with Prop.~\ref{prop:delta-x0},  $A^{4}d_{[1]}(A,q)\chi_0(-iA,iq)$ yields the difference term $\bar\delta(5_2)$, given in Table~\ref{tab:difference}.  
Alternatively, the knot $5_2$ can be expressed as the closure of the knit $e_1\sigma_2^3\sigma_1^{-1}\sigma_2$.

 \vskip 2mm
\noindent{\bf $\circ$}
 Example $4_1$: The figure-8 knot, is the simplest non HZ-factorisable knot. With braid 
$x=(\sigma_1\sigma_2^{-1})^2$, $w=0$, we compute $\chi_{[3]}= 1$ and 
 $\chi_{[21]}(x)=\text{tr}((\rho_{[21]}(\sigma_1)(\rho_{[21]}(\sigma_2))^{-1})^2)= \beta^4+ 2 \beta^2 + 1 + 2\beta^{-2} + \beta^{-4}$, which agrees with $h^{[21]}(4_1;\beta=i q)$ 
and   $\chi_0=   
  (\beta^{-2}+2 + \beta^2)\al^2 +  
  (2\beta^{-3} + 4\beta^{-1} + 4 \beta + 2 \beta^3)\al+\beta^{-4} + 4\beta^{-2}+6 + 4 \beta^2 + \beta^4 + (2\beta^{-3} + 4\beta^{-1} + 4 \beta + 2 \beta^3)\al^{-1}+( \beta^{-2} + 2+\beta^2)\al^{-2} $. 
By (\ref{new}) these yield $KF(4_1;\alpha,z=\beta+\beta^{-1})$ (c.f. \cite{Bar-Natan}). 

\begin{prop}\label{prop:F3E3matrices}
Full twists $F_3^k:=(\sigma_2\sigma_1)^{3k}$ are central elements of the BMW algebra. Their matrix representations  read
\ba\label{r0F3}
&&\hspace{13mm}\rho_{[21]}(F_3^k)=(-1)^k I_2,\nonumber\\
 && \rho_0(F_3^k)=(\rho_0(\sigma_2)\rho_0(\sigma_1))^{3k}=\al^{-2k}I_3,\ea
 where $I_n$ is the  $n\times n$ identity matrix.
For Jucys-Murphy twists 
$E_3:=\sigma_1\sigma_2^2\sigma_1$, the BMW algebra representations are
\ba\label{r2E3}
&&\hspace{9mm}\rho_{[21]}(E_3^k)=(-1)^k \begin{pmatrix} 
  \beta^{2k}&0\\
  -i\sum_{j=0}^{2k-1}\beta^{2(k-j)-1}&\beta^{-2k}
  \end{pmatrix},\nonumber\\
&&\hspace{-7mm}\rho_0(E_3^k)=\al^{-2k}\left(
\begin{array}{ccc}
 \sum_{i=0}^{2k}\beta^{2(k-i)} &
   \sum_{i=0}^{2k-1}\beta^{2(k-i)-1} & 0 \\
 -\sum_{i=0}^{2k-1}\beta^{2(k-i)-1}  &
   -\sum_{i=0}^{2(k-1)}\beta^{2(k-i-1)}  & 0 \\
 x &  y & \al^{2k} \\
\end{array}
\right),
\ea
where we define $x:=\sum_{i=0}^{2k-1}\al^i(\beta^{2k-i}+2\sum_{r=0}^{2k-2-i}\beta^{2(k-1-r)-i}+\beta^{-2k+i})$ and $y:=\sum_{i=0}^{2k-2}\al^i(\beta^{2k-1-i}+2\sum_{r=0}^{2k-3-i}\beta^{2k-3-i-2r}+\beta^{-2k+1+i})$. 
For partial full twists of the form $\sigma_2^{2k}$
\ba\label{r2F2}
&&\hspace{9mm}\rho_{[21]}(\sigma_2^{2k})= \begin{pmatrix} 
  \beta^{-2k}&0\\
  i\sum_{j=0}^{2k-1}\beta^{2(k-j)-1}&\beta^{2k}
  \end{pmatrix},\nonumber\\
&&\hspace{-5mm}\rho_0(\sigma_2^{2k})=\left(
\begin{array}{ccc}
- \sum_{i=0}^{2(k-1)}\beta^{2(k-1-i)} &
  - \sum_{i=0}^{2k-1}\beta^{2(k-i)-1} & 0 \\
 \sum_{i=0}^{2k-1}\beta^{2(k-i)-1}  &
   \sum_{i=0}^{2k}\beta^{2(k-i)}  & 0 \\
 x' &  y' & \al^{-2k} \\
\end{array}
\right),
\ea
in which we define  $x':=\sum_{i=0}^{2k-2}\al^{-1-i}(\beta^{2k-1-i}+2\sum_{r=0}^{2k-3-i}\beta^{2k-3-i-2r}+\beta^{-2k+1+i})$ and $y':=\sum_{i=0}^{2k-1}\al^{-1-i}(\beta^{2k-i}+2\sum_{r=0}^{2k-2-i}\beta^{2(k-1-r)-i}+\beta^{-2k+i})$. Note that $\rho_Q(F_3^k)$, $\rho_Q(E_3^k)$ and $\rho_Q(\sigma_2^{2k})$ commute with each other. 
\end{prop}
\begin{proof}
    By direct computation using (\ref{r0rep}) and (\ref{matrixrepresentation}).
\end{proof}
The different version of the Jucys-Murphy braid $\tilde{E}_3:=\sigma_2\sigma_1^2\sigma_2$ appeared to have simpler (diagonal) representations in the $SU(N)$ case, as compared to $E_3$ \cite{Petrou3}. However, in the $SO(N+1)$ case, $\rho_Q(\tilde{E}_3^k)$ does not simplify, but its  entries only get reshuffled as compared to (\ref{r2E3}). It is noteworthy that 
full twists can be constructed as combinations of partial full twists and Jucys-Murphy twists \cite{Petrou3}. For example, they can be expressed as 
    \be
    F_m=(\sigma_{m-1}\cdots\sigma_2)^{m-1}\otimes E_m=(\sigma_{m-2}\cdots\sigma_1)^{m-1}\otimes\tilde{E}_m,
    \ee
 where $\tilde{E}_m=\sigma_m\sigma_{m-1}\cdots\sigma_2\sigma_1^2\sigma_2\cdots\sigma_m$, or alternatively, as products of Jucys-Murphy twists alone, applied on decreasing number of strands as 
\be
F_m=E_m\otimes E_{m-1}\otimes\cdots\otimes E_2,
\ee
where $E_{i}$ for $i<m$, involves the last $i$ number of strands.

\paragraph{4 strands.} The character expansion for 4 strand knots 
involves Young diagrams with $4,2$ and $0$ boxes and can be written as
\be\label{new4}
{KF(x;\alpha,\beta)=-\alpha^{-w}\left(\chi_0(x)\eta_{0}+ \sum_{Q ,|Q|=2}\chi_{Q}(x)\eta_Q+ \sum_{Q ,|Q|=4}\chi_{Q}(x)\eta_{Q}\right),}
\ee
where $\eta_Q=\frac{d_{Q}(i\al,-i\beta)}{d_{[1]}(i\al,-i\beta)}$, which for $|Q|=0,2$ are given in 
(\ref{KFdim}) while 
for $|Q|=4$ 
they become 
\ba
&&\eta_{[4]}(\alpha,\beta)=\frac{\beta^3(\alpha^2+1)(\alpha \beta+1)(\alpha^2 \beta^4+1)(\alpha \beta^7+1)}{\alpha^3 (\beta^2+1)^3 (\beta^2-1)^2 (\beta^4+1) (-1+ \beta^2- \beta^4)}\nonumber\\
&&\eta_{[31]}(\alpha,\beta)=\frac{\beta(\alpha^2+1)(\alpha+\beta)(\alpha^2 \beta^4+1)(\alpha+ \beta^3)(\alpha \beta^5-1)}{\alpha^3 (\beta^2+1)^3 (\alpha \beta-1)(\beta^2-1)^2 (\beta^4+1)}\nonumber\\
&&\eta_{[22]}(\alpha,\beta)=\frac{(\alpha+\beta^3)(\alpha \beta^3+1)(\alpha^2+ \beta^4)(\alpha^2\beta^4+1)}{\alpha^3 (\beta^2+1)^3 (\beta^2-1)^2 (-1+\beta^2-\beta^4)}\nonumber\\
&&\eta_{[211]}(\alpha,\beta)=\frac{\beta(\alpha^2+1) (\alpha \beta+1)(\alpha^2 + \beta^4))(\alpha \beta^3+1)(\alpha- \beta^5)}{\alpha^3 (\beta^2+1)^3 (\alpha - \beta)(\beta^2-1)^2 (\beta^4+1)}\nonumber\\
&&\eta_{[1111]}(\alpha,\beta)=\frac{\beta^3(\alpha^2+1)(\alpha +\beta)(\alpha ^2+ \beta^4)(\alpha+ \beta^7)}{\alpha^3 (\beta^2+1)^3 (\beta^2-1)^2 (\beta^4+1) (-1+ \beta^2- \beta^4)}.
\ea

As before, when $|Q|=4$ the coefficients  can be obtained from the Racah coefficients  as $\chi_Q(\beta)=i^wh^Q(q)|_{q=-i\beta}$, 
and hence can be computed using 
the $R-$matrices in \cite{Petrou3}. 
The
matrix representations corresponding to Young diagrams with 2 boxes are 6 dimensional (c.f. Fig.~\ref{fig:Bratelli}), for which explicit expressions,
 as given  by Birman and Wenzl in 
\cite{Birman}, are 
\ba\label{4strandsigma}
\rho_{[11]}(\sigma_1)&=& \begin{pmatrix} 
  \beta^{-1}&0&0&0&0&0\\
0&0&0&0&-1&0\\
0&0&0&0&0&-1\\
0&0&0&\alpha^{-1}&\alpha^{-1}z&\alpha^{-1}\beta z\\
 0&1&0&0&z&0\\
 0&0&1&0&0&z
 \end{pmatrix},\nonumber\\
 \rho_{[11]}( \sigma_2)&=& \begin{pmatrix} 
  0&0&-1&0&0&0\\
0&\alpha^{-1}&\alpha^{-1}z&z&0&0\\
1&0&z&0&0&0\\
0&0&0&z&1&0\\
 0&0&0&-1&0&0\\
 0&0&0&0&0&\beta^{-1}
 \end{pmatrix},\nonumber\\
 \rho_{[11]}(\sigma_3)&=& \begin{pmatrix} 
  \alpha^{-1}&z&0&0&\beta^{-1}z&0\\
0&z&1&0&0&0\\
0&-1&0&0&0&0\\
0&0&0&\beta^{-1}&0&0\\
 0&0&0&0&z&1\\
 0&0&0&0&-1&0
 \end{pmatrix}.
 \ea
Notably, these matrix representations 
consist of 3 block matrices of dimensions 3,2 and 1, as suggested by the Bratelli diagram.  
The blocks, consisting of    a $3\times3$ matrix in the upper left corner of $\rho_{[11]}(\sigma_3)$ and $\rho_{[11]}(\sigma_2)$ are  the same as $\rho_0 (\sigma_1)$  
and   $\rho_0(\sigma_2)$ in footnote~\ref{foot:r0s2}, respectively. These matrices satisfy 
$\text{det}(\rho_{[11]}(\sigma_i))=\alpha^{-1}\beta^{-1}$, $\text{det}(\rho_{[11]}(\sigma_i)^{-1})= \alpha \beta$, $\text{tr}(\rho_{[11]}(\sigma_i))=\alpha^{-1}+\beta^{-1}+ 2 z$, for $i=1,2,3$,  and $\text{tr}(\rho_Q(\sigma_3)\rho_Q(\sigma_2)\rho_Q(\sigma_1))=0.$
From the defining relation (ii) $\sigma_i+\sigma_i^{-1}= z (1+ e_i)$,
the matrices $\rho_{[11]}(e_i)$ become
\ba
&&\hspace{-6mm} \rho_{[11]}(e_1)= \begin{pmatrix} 
  0&0&0&0&0&0\\
0&0&0&0&0&0\\
0&0&0&0&0&0\\
0&1&\beta&\mu&\alpha^{-1}&\alpha^{-1}\beta \\
 0&0&0&0&0&0\\
 0&0&0&0&0&0
 \end{pmatrix},\hskip1mm
 \rho_{[11]}(e_2)= \begin{pmatrix} 
  0&0&0&0&0&0\\
1&\mu&\alpha^{-1}&1&\alpha&0\\
0&0&0&0&0&0\\
0&0&0&0&0&0 \\
 0&0&0&0&0&0\\
 0&0&0&0&0&0
 \end{pmatrix}\nonumber\\
 &&\hspace{20mm}\rho_{[11]}(e_3)= \begin{pmatrix} 
  \mu&1&\alpha&0&\beta^{-1}&\alpha \beta^{-1}\\
0&0&0&0&0&0\\
0&0&0&0&0&0\\
0&0&0&0&0&0 \\
 0&0&0&0&0&0\\
 0&0&0&0&0&0
 \end{pmatrix}
 \ea
 where $\mu=(\alpha+\alpha^{-1})/z-1$. These also satisfy the braid relations, the idempotent relation (\ref{idempotent}) and, in agreement with $KF(\bigcirc^n)=\mu^n$,  $
\text{tr}(\rho_{[11]}(e_i)^n)= \mu^n$, for $ n\in \mathbb{Z}$. Moreover, they satisfy the further properties  $\text{tr}(e_1e_2)=\text{tr}(e_2 e_3)=1$ and $\text{tr}(e_1 \sigma_1)= \alpha^{-1}\mu, $ and $
\text{tr}(e_1 \sigma_1^{-1})= \alpha \mu.$

 The representations for  $\rho_{[2]}(\sigma_i)$ are obtained from the matrices $ \rho_{[11]}(\sigma_i)$ by $\beta\mapsto \beta^{-1}$ (or equivalently $q\mapsto-q^{-1}$).
The representations corresponding to the  Young diagram with 0 boxes $\rho_0(\sigma_i)$ are 3 dimensional, and they are given by (\ref{r0rep}) for $i=1,2$, while $\rho_0(\sigma_3)=\rho_0(\sigma_1)$.

 \begin{prop} For a 4-strand knit  the extra term in the character expansion of the unnormalised Dubrovnik polynomial $\bar{Y}(A,q)=A^{-w}\sum_Qh^Q(q)d_Q(A,q)+\bar\delta(A,q)$ reads
 \be\label{delta4}
 \bar\delta(A,q)=(iA)^{-w}\left( \chi_0(-iA,iq)+ \sum_{Q , |Q|=2}\chi_{Q}(-iA,iq)d_{Q}(A,q)\right)
 \ee 
  \end{prop}
  \begin{proof}
      The proof is the same as in the 3-strand case in Prop.~\ref{prop:delta-x0}, but now with the parities of $c$ and $w$ matching.
  \end{proof}

\vskip2mm
\noindent{$\circ$} Example  $T(4,n)$. For the braid $(\sigma_3\sigma_2\sigma_1)^n$ with $n$ odd, we compute
$\chi_{[11]}=\chi_{[2]}=0$ and $\chi_0=-\al^{-n}$ yielding $\bar\delta$ in (\ref{T(4,n)}). 
\vskip2mm
\noindent{$\circ$} Example  $10_{128}$.  With braid $\sigma_1^3\sigma_2\sigma_1^2\sigma_2^2\sigma_3\sigma_2^{-1}\sigma_3$
we find $\chi_0=-\al^{-3}$ and
\ba
\chi_{[11]}\hspace{-2mm}&=&\hspace{-2mm}\al (\beta^{-6}+\beta^{-4})-\al^{-1}(\beta^{-6}+3\beta^{-4}+3\beta^{-2}+1)\nonumber\\
 &&\hspace{-9mm}+\beta^{-7}+\beta^{-5}-\beta^{-3}-\beta^{-1}+\al^{-2}(-\beta^{-7}-2\beta^{-5}-\beta^{-3}+\beta^{-1}+2\beta+\beta^{3})\nonumber\\
 &&\hspace{-9mm}+\al^{-3}(\beta^{-4}+2\beta^{-2}+2+\beta^{2}),
    \ea
    while $\chi_{[2]}(\al,\beta)=\chi_{[11]}(\al,\beta^{-1})$.
By (\ref{delta4}), these give precisely $\bar \delta(10_{128})$ in (\ref{10_128}). 
\begin{prop}\label{prop:F4E4matrices}
Full twists $F_4^k:=(\sigma_3\sigma_2\sigma_1)^{4k}$ are, again, central in the BMW algebra. Their matrix representations for $|Q|=4$ read \cite{Petrou3}
\be
\rho_{[31]}(F_4^k)=\beta^{4k}I_3,\;\;\rho_{[22]}(F_4^k)=I_2,
\ee
while for $|Q|=2,0$ 
\be\label{r11F4}
\rho_{[11]}(F_4^k)=(\al\beta)^{-2k} I_6, \;\;\rho_{[2]}(F_4^k)=\al^{-2k}\beta^{2k} \mathbb{I}_6, \;\;\rho_0(F_4^k)=\al^{-4k}I_3,
\ee
 with $I_n$  the  $n\times n$ identity matrix.
For Jucys-Murphy twists $\tilde{E}_4:=\sigma_3\sigma_2\sigma_1^2\sigma_2\sigma_3$, we compute
\be\label{E4rep}
\rho_{[31]}(\tilde{E}_4)=\begin{pmatrix} -\beta ^{-2}& \\
&\beta^4 I_2\end{pmatrix},\;\;\rho_{[22]}(\tilde{E}_4^k)=I_2, \;\;\rho_0(\tilde{E}_4^k)=\al^{-2k}I_3,
\ee
while for  $Q=[31],[11]$ they  are non-diagonal and  lengthy and hence we omit them here for simplicity. For the partial full twist $F_3:=(\sigma_2\sigma_1)^3$, again omitting the non-diagonal $6\times 6$ matrices, we find
     \be
\rho_{[31]}(F_3)=\begin{pmatrix} -\beta ^{6}& \\
& I_2\end{pmatrix},\;\; \rho_{[22]}(F_3)=I_2, \rho_0(F_3)=\al^{-2}I_3.
    \ee
\end{prop}
\begin{proof}
    For $Q$ s.t. $|Q|=4$ these are results from \cite{Petrou3}  (with $\beta=iq$) and for the remaining $Q$, by direct computation using (\ref{r0rep}) and (\ref{4strandsigma}).
\end{proof}

\section{The HOMFLY--PT/Kauffman relation and HZ factorisability 
}\label{sec:H/KF}
 The HOMFLY--PT  and Kauffman polynomials are in general not related, with the latter usually being much more lengthy and complicated. However, in \cite{Perez} it was found using gauge theoretic techniques, that for torus knots these two polynomials can be  related by\footnote{In \cite{Stevan} this HOMFLY--PT/Kauffman relation  for torus knots was investigated  for representations beyond 
the fundamental.}
 \ba\label{KFhat}
&&\hspace{-6mm}H(\mathcal{K};A,z)=\widetilde{KF}_{\rm even}(\mathcal{K};A,z)-\frac{z}{A-A^{-1}}\widetilde{KF}_{\rm odd}(\mathcal{K};A,z)=:\widehat{KF}(\mathcal{K};A,z)\nonumber\\
&&\hspace{11mm}\widetilde{KF}_{\rm even/odd}(A,z):= \frac{1}{2}\left(\widetilde{KF}(iA,iz)\pm\widetilde{KF}(iA,-iz)\right).
 \ea
The authors then naturally posed the question: Does this relation holds for other  knots or links beyond the torus family? An affirmative answer was given by Theorem 4.1 of \cite{Petrou2}, which proves that this relation holds for several hyperbolic families of knots, that were found through HZ-factorisability properties of the HOMFLY--PT polynomial. In the present paper, we extend this result by viewing this relation under the lens of the character expansions, outlined in the previous sections.
Peculiarly, in the context of topological string theory this relation is connected to the vanishing of  BPS degeneracy (i.e. the number of BPS states) with two crosscaps \cite{Petrou2}.  
 
The relation $H(A,z)=\widehat{KF}(\mathcal{K};A,z)$ can be  expressed in terms of the unnormalised Dubrovnik polynomial as follows. Multiplying by an overall factor $1-(\frac{A-A^{-1}}{z})^2$ on both sides of (\ref{KFhat}) and using $\widetilde{\overline{KF}}=(\frac{A-A^{-1}}{z}-1 )\widetilde{KF}$ and $\overline{H}=\frac{A-A^{-1}}{z}H$, yields
 \ba
&&\hspace{-7mm}\left(\frac{z}{A-A^{-1}}-\frac{A-A^{-1}}{z}\right)\overline{H}(A,z)\nonumber\\
&&\hspace{-5mm} =\frac{1}{2}\left(-\left(\frac{A-A^{-1}}{z}+1\right)\widetilde{\overline{KF}}(A,z)+\left(\frac{A-A^{-1}}{z}-1\right)\widetilde{\overline{KF}}(A,-z)\right)\nonumber\\
 &&\hspace{-5mm} -\frac{z}{2(A-A^{-1})}\left(-\left(\frac{A-A^{-1}}{z}+1\right)\widetilde{\overline{KF}}(A,z)-\left(\frac{A-A^{-1}}{z}-1\right)\widetilde{\overline{KF}}(A,-z)\right),\nonumber
 \ea
  which simplifies to
 \be\label{HKFbar}
 \bar{H}(A,z)=\frac{1}{2}(\widetilde{\overline{KF}}(A,z)-\widetilde{\overline{KF}}(A,-z)).
 \ee
 Using that $\widetilde{KF}(A,z)=KF(iA,iz)$ and  (\ref{KF-Y}), this can be written as
 \be\label{homfly-Dubrovnik}
\bar{H}(A,q)=\frac{1}{2}\left(\bar{Y}(-A,-z)-\bar{Y}(-A,z)\right).
\ee

\begin{prop}\label{HKF3x0}
    The HOMFLY--PT/Kauffman relation (\ref{KFhat}), or equivalently (\ref{homfly-Dubrovnik}), for a knot $\mathcal{K}$ with a 3-strand braid representative can be expressed in terms of $\chi_0$ as
\be\label{x0HKF}
\widetilde{\chi_0}_{\rm odd}(\mathcal{K};A,q)=\frac{A-A^{-1}}{q-q^{-1}}\widetilde{\chi_0}_{\rm even}(\mathcal{K};A,q),
\ee
where $\widetilde{\chi_0}_{\rm even/odd}(A,q):=\frac{1}{2}\left(\chi_0(iA,iq)\pm\chi_0(iA,-iq)\right)$. 
\end{prop}
\begin{proof}
    Using  (\ref{Thm5.1}) together with the fact that $h^Q(-q)=h^Q(q)$, since it only depends on even powers of $q$, and that $d_Q(-A,-q)=d_Q(A,q)$, the RHS   of (\ref{homfly-Dubrovnik}) becomes
    \ba
   &&\hspace{-7mm} \frac{1}{2}(-1)^wA^{-w}\sum_Qh^Q\left(d_Q(A,q)-d_Q(-A,q)\right)\\\nonumber
   &&\hspace{-2mm}\pm\frac{1}{2}(-1)^wA^{-w}\left(\left(1+\frac{A-A^{-1}}{q-q^{-1}}\right)\chi_0(iA,-iq)-\left(1-\frac{A-A^{-1}}{q-q^{-1}}\right)\chi_0(iA,iq)\right)
    \ea
    Using (\ref{3strand1}) it is easy to show $d_Q(A,q)-d_Q(-A,q)=2S_Q(A,q)$ and since $w$ is always even for 3-strand \emph{knots}, the term in the first line exactly cancels with the l.h.s of (\ref{homfly-Dubrovnik}), which is expressed by the character expansion as $\bar H=A^{-w}\sum_Q h^Q S_Q$ \cite{Petrou3}. The remaining terms are vanishing when (\ref{x0HKF}) holds.
\end{proof}

\begin{cor}
    \label{thm:FTeffect} Concatenation by full twists $F_3^k$ does not affect the HOMFLY--PT/Kauffman relation.
\end{cor}
\begin{proof}
 Using (\ref{r0F3}) and the trace identity $\text{tr}(cM)=c\text{tr}(M)$ for  a scalar $c$ and  a matrix $M$, we compute $\chi_0(\mathcal{K}\otimes F_3^k)=\text{tr}(\rho_0(\mathcal{K})\al^{-2k}I_3)=\al^{-2k}\chi_0(\mathcal{K})$. Hence, if  $\chi_0(\mathcal{K})$ satisfies (\ref{x0HKF}), so does $\chi_0(\mathcal{K}\otimes F_3^k)$,  and by  Prop.~\ref{HKF3x0} the HOMFLY--PT/Kauffman relation is satisfied.
\end{proof} 
It is noteworthy that since full twists $F_3$ correspond to Dehn twists, their concatenation does not depend on the specific braid word chosen for $\mathcal{K}$. This is also reflected by the fact that $F_3$ are central elements in the BMW algebra (c.f. Prop.~\ref{prop:F3E3matrices}). 
An analogue of Cor.~\ref{thm:FTeffect} for Jucys-Murphy twists $E_3$, however, is not so straight forward, because of the braid word-dependence of the resulting knot, obtained after concatenation with $E_3$. In  particular, for Prop.~\ref{HKF3x0} to hold after Jucys-Murphy twists are attached, the chosen 3-strand braid representative of a knot $\mathcal{K}$ needs  to satisfy not only (\ref{x0HKF}), but also some additional conditions on the off-diagonal elements of the first and second rows   of  $\rho_0(\mathcal{K})$, which  in general will affect  $\chi_0(\mathcal{K}\otimes E_3)$.

It was conjectured in \cite{Petrou2} that
the above mentioned relation between the HOMFLY--PT and Kauffman polynomials of a knot, occurs
 if and only if the  HZ transform of the HOMFLY--PT polynomial is factorisable, that is
  \be\label{conj}
 { Z(\mathcal{K}) = \frac{\lambda \prod_{i=0}^{m-2}(1-\lambda q^{\alpha_i})}{\prod_{i=0}^m (1-\lambda q^{\beta_i})}\iff H (\mathcal{K}) = \widehat {KF}(\mathcal{K})  }.
 \ee
\begin{conj}\label{conj:3strand}
    The conjectured relation (\ref{conj}) between HZ factorisability and the HOMFLY--PT/Kauffman relation,   for a knot that is a closed 3-strand  braid $\mathbf{b}$, can be stated in terms of  characters of the BMW algebra representations as 
 \be\label{3conjecture}
\hspace{-1mm}
\boxed{\chi_{[21]}(iq)
=\pm \frac{q^{\gamma}+q^{-\gamma}}{q+q^{-1}}\iff{\chi_0}_{\rm odd}(iA,iq)=\frac{A-A^{-1}}{q-q^{-1}}{\chi_0}_{\rm even}(iA,iq),}
\ee
where $\widetilde{\chi_0}_{\rm even/odd}(A,q):=\frac{1}{2}\left(\chi_0(iA,iq)\pm\chi_0(iA,-iq)\right)$   
and $\gamma(\mathbf{b})\in \mathbb{Z}_{\rm even}$. 
\end{conj}
The condition on the left hand side 
was shown to imply HZ-factorisability in \cite{Petrou3}. It was also shown that $\chi_{[21]}$ is invariant under full twists, partial full twist 
and Jucys-Murphy twists $\tilde{E}_3$. 
However, using this complex BMW representation, the effect of Jucys-Murphy twist seems to be more subtle. In particular, if for a  knot $\mathcal{K}$ with $  \rho_{[21]}(\mathcal{K})= \begin{pmatrix} 
  x_{11}&x_{12}\\
  x_{21}&x_{22}
  \end{pmatrix}$, s.t. $\chi_{[21]}$ 
  satisfies the condition on the l.h.s,
   $\chi_{[21]}(\mathcal{K}\otimes E_3^k)=(-1)^k(x_{11}\beta^{2k}-ix_{12}\sum_{j=0}^{2k-1}\beta^{2(k-j)-1}+x_{22}\beta^{-2k})$, may still satisfy or not satisfy it based on the off diagonal term $x_{12}$. 
  In fact, given a specific initial braid word, HZ factorisability will depend on which version of Jucys-Muprhy twist is attached to it. This explains the braid-word-dependence under Jucys-Murphy concatenation, that was observed  in \cite{Petrou2}.  

  As mentioned in \cite{Petrou3}, the most general 3-strand HZ factorisable family can be expressed as the closed braid $\sigma_1\sigma_2^{\pm(1+2j)}\otimes F_3^l\otimes E_3^k$, 
  with $j,k,l\in\mathbb{Z}$, or  equivalently  as $\sigma_2\sigma_1\otimes F_2^j\otimes F_3^l\otimes \tilde{E}_3^k=:\mathcal{K}^{(3)}_{j,k,l}$, where $F_2=\sigma_1^2$ is the partial full twist. Here shall here denote this family by $\mathcal{K}^{\pm}_{j,k,l}$, explicitly indicating positive or negative partial full twists. 
\begin{thm}\label{thm:Kjkl,H-KF} The 3-strand family of knots $\mathcal{K}^{\pm}_{j,k,l}:=\sigma_1\sigma_2^{\pm(1+2j)}\otimes F_3^l\otimes E_3^k$ is HZ-factorisable and satisfies the HOMFLY--PT/Kauffman relation $H=\widehat{KF}$  (\ref{KFhat}), or equivalently Prop.~\ref{HKF3x0}. 
\end{thm}
\begin{proof}
Using the matrix representations $\rho_{[21]}(F_3^k)$ and $\rho_{[21]}(E_3^k)$ in Prop.~\ref{prop:F3E3matrices}, we compute  
  \ba
&&\hspace{-10mm}\chi_{[21]}(\mathcal{K}^{-}_{j,k,l})=\\\nonumber
&&\hspace{-2mm}(-1)^{k+l}\left(\beta^{2(j+k+1)}+\sum_{i=0}^{2j}\beta^{2(j+k-i)}+\sum_{i=0}^{2k-1}\beta^{2(k-j-i-1)}+\beta^{-2(j+k+1)}\right)
  \ea
  and
  \be
  \chi_{[21]}(\mathcal{K}^{+}_{j,k,l})=(-1)^{k+l}\left(\beta^{2(k-j)}-\sum_{i=0}^{2j}\beta^{2(j+k-i)}+\sum_{i=0}^{2k-1}\beta^{2(k+j-i)}+\beta^{2(j-k)}\right).
  \ee
  Both of the above, which  agree\footnote{Up to an overall factor $(-1)^l$.} 
  with the expression computed in \cite{Petrou3},  satisfy the HZ-factorisability condition on the left hand side of (\ref{3conjecture}). 
    Moreover, again by direct computation using Prop.~\ref{prop:F3E3matrices}, we find
\ba\label{chi0kjl}
&&\hspace{-7mm}\chi_0(\mathcal{K}^{-}_{j,k,l})=\al^{-2(k+l)-1}\sum_{i=0}^{2k+2j+1}\beta^{2(k+j-i)+1}\\\nonumber
&&\hspace{-5mm}+\sum_{i=0}^{2(k+j)+1}\al^{-2(k+l)+i}\left(\beta^{2(k+j+1)-i}+\beta^{-2(k+j+1)+i}+2\sum_{r=0}^{2(k+j)-i}\beta^{2(k+j-r)-i}\right),
\ea
while $\chi_0(\mathcal{K}^{+}_{j,k,l})$ is given by the same formula with $j\mapsto -j$ when $k> j$ or by $\chi_0(\mathcal{K}^{+}_{j,k,l})=0$ when 
$k\leq j$, which correspond to 3-strand torus knots (e.g. $\mathcal{K}^{+}_{2,1,0}=T(3,5)$).
The expression in (\ref{chi0kjl}) can be readily split in odd and even parts $\widetilde{\chi_0}_{\rm even/odd}(A,q)=\frac{1}{2}(\chi_0(iA,iq)\pm\chi_0(iA,-iq))$, as 
\ba\label{chi0odd}
\widetilde{\chi_0}_{\rm odd}(\mathcal{K}^-_{j,k,l})&=&(iA)^{-2(k+l)-1}\sum_{i'=0}^{2k+2j+1}(iq)^{2(k+j-i')+1}\nonumber\\
&&+\sum_{i=1',{\rm odd}}^{2(k+j)+1}(iA)^{-2(k+l)+i'}\bigg((iq)^{2(k+j+1)-i'}\nonumber\\
&&+2\sum_{r=0}^{2(k+j)-i'}(iq)^{2(k+j-r)-i'}+(iq)^{-2(k+j+1)+i'}\bigg),
\ea
\ba\label{chi0even}
\widetilde{\chi_0}_{\rm even}(\mathcal{K}^-_{j,k,l})\hspace{-2mm}&=&\hspace{-4mm}\sum_{i'=0,{\rm even}}^{2(k+j)}(iA)^{-2(k+l)+i'}\bigg((iq)^{2(k+j+1)-i'}\nonumber\\
&&+2\sum_{r=0}^{2(k+j)-i'}(iq)^{2(k+j-r)-i'}+(iq)^{-2(k+j+1)+i'}\bigg).
\ea
From (\ref{chi0even}), noticing that the term in the bracket is divisible by $q-q^{-1}$, we compute
\ba
&&\hspace{-8mm}\frac{A-A^{-1}}{q-q^{-1}}\widetilde{\chi_0}_{\rm even}(\mathcal{K}^-_{j,k,l})\nonumber\\\nonumber
&&\hspace{-3mm}=\sum_{i'=0,{\rm even}}^{2(k+j)}\left((iA)^{-2(k+l)+1+i'}+(iA)^{-2(k+l)-1+i'}\right)\sum_{r=0}^{2(k+j)+1-i'}(iq)^{2(k+j)+1-i'-2r}\nonumber\\
&&\hspace{-3mm}=(iA)^{-2(k+l)-1}\sum_{i'=0}^{2(k+j)+1}(iq)^{2(k+j-i')+1}\nonumber\\
&&\hspace{20mm}+\sum_{i'=2,{\rm even}}^{2(k+j)}(iA)^{-2(k+l)-1+i'}\sum_{r=0}^{2(k+j)+1-i'}(iq)^{2(k+j-r)+1-i'}\nonumber\\
&&\hspace{20mm}+\sum_{i'=0,{\rm even}}^{2(k+j)}(iA)^{-2(k+l)+1+i'}\sum_{r=0}^{2(k+j)+1-i'}(iq)^{2(k+j-r)+1-i'}.
\ea
After relabeling $i''=-1+i'$ and $i'''=1+i'$ in the first sums of the last two lines, respectively, while giving them the same upper bound $2(k+j)+1$  (this is allowed since  $\sum_{r=0}^{-1}$ adds a vanishing contribution), it can be  shown with some simple algebraic manipulations that this expression becomes exactly equal to the odd part in (\ref{chi0odd}) and hence $\chi_0(\mathcal{K}^-_{j,k,l})$ satisfies (\ref{x0HKF}). The proof is the same for $\mathcal{K}^+_{j,k,l}$ when $k>j$, and trivial otherwise.
\end{proof}

 Hence, Conjecture~\ref{conj:3strand} is valid for all 3-strand HZ-factorisable knots known to us, which belong to the family $\mathcal{K}^{(3)}_{j,k,l}$. Finding a complete proof of the conjecture in the 3-strand case, would amount to showing that a 3-strand knot satisfies either of the conditions in the conjecture if and only if it belongs to this family. 

In the case of 4 strands a sufficient\footnote{\label{foot:L10n42}This condition is not necessary, as the link counter example $L10n_{42}\{1\}$ shows \cite{Petrou3}.} condition for HZ factorisability was found in \cite{Petrou3} to be 
\ba\label{4fact}
&&{\rm (i)} \;h^{[22]}=0,\nonumber\\
&&{\rm (ii)}\;h^{[31]}=-\frac{q^{\gamma_1}+q^{\gamma_2}+q^{\gamma_3}}{q^2+1+q^{-2}},
\ea
where $\gamma_{i}$ are odd integers  satisfying $\gamma_1+\gamma_2+\gamma_3=w$. Moreover, the condition for the HOMFLY--PT/Kauffman relation can be stated as follows. 
\begin{prop}
    For knots with 4-strand braid representatives, the condition for the  HOMFLY--PT/Kauffman  relation is expressed in terms of characters of the BMW algebra as 
\be\label{HKF4strands}
\sum_{Q,|Q|=2}\left(\chi_Q(iA,-iq)d_Q(A,q)-\chi_Q(iA,iq)d_Q(-A,q)\right)=-{\chi_{0}}_{\rm odd}(iA,iq),
\ee
where ${\chi_{0}}_{\rm odd}(iA,iq):=\chi_0(iA,iq)-\chi_0(iA,-iq).$ 
\end{prop}
\begin{proof}
    The proof is similar to the one for 3-strands in Prop.~\ref{HKF3x0}, with the extra terms now given by (\ref{delta4}). Moreover, one should notice that for  4-strand knots 
    the writhe $w$ is odd, $h^Q(-q)=-h^Q(q)$ and $d_Q(A,q)+d_Q(-A,q)=2S_Q$.
\end{proof}
\noindent{}The condition (\ref{HKF4strands}) can be easily checked to be verified in the 4-strand HZ-factorisable examples of the previous sections, such as for $10_{128}$, $10_{132}$, $12n_{318}$, and that it fails to hold for the simplest 4-strand knot $6_1$, which is not HZ-factorisable. 

Two general  4-strand families that  satisfy the HZ-factorisability conditions (i), (ii) in (\ref{4fact}), generated by $F_3=(\sigma_2\sigma_1)^{3}$, $\tilde{E}_4=\sigma_3\sigma_2\sigma_1^2\sigma_2\sigma_3$ and $F_4=(\sigma_3\sigma_2\sigma_1)^{4}$   are  $\mathcal{K}_{j,k,l}^{(4)}:=\sigma_3\sigma_2\sigma_1\otimes F_3^j\otimes F_4^l\otimes \tilde{E}_4^k$, 
which includes $10_{132}=\mathcal{K}^{(4)}_{-2,0,1}$ and  torus knots when $j=k$ and $j=k+1$; 
and $\sigma_3\sigma_1^{-1}\sigma_2^{-2}\sigma_1\sigma_2^{-1}\sigma_1^{-1}\otimes F_3^j\otimes F_4^l\otimes \tilde{E}_4^k$, which includes $10_{128}$ and $12n_{318}$ at $j=0,k=0,l=1$, and $j=1,k=0,l=1$, respectively. Their Racah coefficients were computed\footnote{As remarked in \cite{Petrou3}, it is worth mentioning the computational efficiency of the character expansion as opposed to the skein relation. Via characters the complexity parameter is the number of strands rather than the number of crossings (e.g. the computation for 4-strand knots with $\sim70$ crossings can be dramatically decreased to a couple of minutes).}  using  the real $R$-matrices in \cite{Petrou3}. 
It can be shown by direct computation that the condition (\ref{HKF4strands}) holds for many of their members. Proving that they satisfy the HOMFLY--PT/Kauffman relation more generally may be achieved by a similar, but significantly more lengthy computation,  as the one that lead to the proof of Theorem~\ref{thm:Kjkl,H-KF}.

Similar to the above considerations apply in the case of two component links for which the HOMFLY--PT/Kauffman relation, in terms of $\widehat{\widehat{KF}}(\mathcal{L}):=\frac{z}{A-A^{-1}}\widetilde{KF}_{\rm even}(\mathcal{L})-\widetilde{KF}_{\rm odd}(\mathcal{L})$ and  the linking number ${\rm lk}(\mathcal{L})$,  reads
 \be\label{linkcriterionA}
H(\mathcal{L};A,z)+\widehat{\widehat{KF}}(\mathcal{L};A,z) = -\frac{zA^{-4lk(\mathcal{L})}}{A-A^{-1}}.
\ee
 In the 3-strand case, this may be expressed in terms of $\chi_0$ as
 \be\label{x0HKFLink}
\boxed{\widetilde{\chi_0}_{\rm even}(\mathcal{L};A,q)-\frac{A-A^{-1}}{q-q^{-1}}\widetilde{\chi_0}_{\rm odd}(\mathcal{L};A,q)=-i^wA^{-4 lk(\mathcal{L})+w}.}
\ee
The validity of this relation has been proven for the general family  of HZ-facrorisable, 2-component links $\mathcal{L}^{(3)}_{j,k,l}:=\sigma_2\otimes F_2^j\otimes F_3^l\otimes \tilde{E}_3^k$ in \cite{PetrouPHD}.

In the 4-strand case, however,  two-component link $L10n_{42}\{1\}$ that admits HZ factorisability provides  a counter example to the left implication of Conjecture~(\ref{conjINTRO}). As pointed out in \cite{Petrou2,Petrou3}, this link is exceptional in the sense that it can not be constructed by full twists and/or Jucys-Murphy twists applied to an HZ-factorisable base braid.


\section{Summary and discussion}
In the present article, we have outlined how  the $SO(N+1)$ character expansion for the  Kauffman polynomial can be used to 
articulate explicit conditions for  the HOMFLY--PT/Kauffman relation in terms of characters of the BMW algebra. Generalising the results of \cite{Petrou2}, in Theorem~\ref{thm:Kjkl,H-KF}  we have established  this relation for a general  family of HZ-factorisable 3-strand knots, which hyperbolically extends torus knots.  This result provides important evidence for  the validity of our conjectured equivalence between HZ-factorisability and the HOMFLY--PT/Kauffman relation in the 3-strand case. 

The aforementioned family  is constructed via three braid operations: full twists, partial full twists and  Jucys-Murphy twists. In fact, these  braid operations  provided an essential tool for the construction of hyperbolic families of HZ-factorisable knots with arbitrary number of strands in \cite{Petrou3}, due to their diagonal and commuting $R-$matrix representations. They also play an important role in the representation theory of the BMW algebra, as their combinations form commutative elements. However, 
 their direct effects on the  conditions  for the    HOMFLY--PT/Kauffman relation are not so straight forward, and hence proving the conjecture more generally remains an open problem. 

The Racah coefficients (6-j symbols), which appear in both the  $SO(N+1)$ and $SU(N)$ character expansions, 
have important applications in several areas of physics. These include 
the Potts (Ising) spin model in statistical mechanics \cite{Kauffman-Lins},  
  quantum computation (quantum gates) \cite{Kauffmangate}, and BPS invariants in  topological string theory
\cite{Petrou2,Bouchard,Mironov}.  The latter  are related to  topological excitations on the boundary in  spin moduli curves. In particular, the BPS invariants in a Calabi-Yau 3-fold are conjectured to be equivalent to the Gromov-Witten invariants
 \cite{Gopakumar}, which reduce to the intersection numbers on Riemann surfaces in two dimensions \cite{HikamiD}.   The  $SO(N+1)$ gauge symmetry group 
arises from  D-branes in the presence of an orientifold, and  becomes a useful tool for the interpretation of  BPS invariants  with boundaries (punctures) \cite{Bouchard}. Finally, the Racah coefficients 
may also play a significant role in non-perturbative  quantum gravity \cite{Kauffmangravity}. 
We hope that this study of the $SO(N+1)$ character expansion, in relation to $SU(N)$, will be useful for  such interesting applications in the future. 

\paragraph{Acknowledgments.}
We are grateful to Louis Kauffman for insightful discussions and suggestions and to Jun Murakami for explaining to us the higher strand expansions. We thank Juan Luis Araujo Abranches
 for his enlightening comments and to Reiko Toriumi for her continuous support. We are grateful to Tancredi Schettini Gherardini who pointed out a miscalculation of a 4-strand counter example to our conjecture, that was falsely claimed in a previous version of this article.
This work is supported by OIST funding  and by the collaboration fund  between the University of Tokyo and OIST.

\paragraph{Declaration of competing interest.} The authors have no competing interest to disclose. 

\addcontentsline{toc}{section}{Appendix A. Quantum dimensions of $SO(N+1)$}
\section*{Appendix A. Quantum dimensions of $SO(N+1)$}

For the even orthogonal Lie group $SO(2n)$, the character of the highest weight irreducible representation (irrep) $\Gamma_Q$, is given by the Weyl's character formula and 
can be expressed by the ratio of determinants as  \cite{Fulton}
\be\label{SO(2n)}
{\rm Char}\;\Gamma_{Q}= \frac{\text{det}|x_j^{l_i}+x_j^{-l_i}|+{\text{det}}|x_j^{l_i}-x_j^{-l_i}|}{{\text{det}}|x_j^{m_i}+x_j^{-m_i}|}.
\ee
Here $l_i= r_i+n-i$, $m_i=n-i$ for $i\in\{1,..,n\}$, where $r_i$ are the number of boxes in the $i^{\rm th}$ row of the Young diagram $Q$.
For instance, the Young diagram $Q=[2]$, has 2 boxes in the first row and hence  $(r_1,r_2,..,r_n)=(2,0,..,0)$. For $n=2$ the corresponding
 character  is 
\ba\label{charExample}
{\rm Char}\;\Gamma_{[2]}&=&\text{det} \begin{pmatrix} x_1^3 + x_1^{-3}&2\\
x_2^3+ x_2^{-3}&2\end{pmatrix}/
\text{det}\begin{pmatrix} 
x_1+x_1^{-1}&2\\
x_2+ x_2^{-1}& 2\\
\end{pmatrix} \nonumber\\
&=& x_1^2+x_1^{-2}+x_2^2+x_2^{-2}+(x_1+x_1^{-1})(x_2+x_2^{-1})+1.
\ea
The  dimension of this representation can be obtained by setting $x_1=x_2=1$ in (\ref{charExample}), which yields ${\rm dim}\;\Gamma_{[2]}=9$ for $SO(4)$. 

The (classical) dimension for highest weight irreps of $SO(2n)$ can  in fact be obtained directly by the 
formula (see e.g. p.410 of \cite{Fulton})  
\be\label{dimFormula}
\text{dim} \hskip1mm \Gamma_Q=\prod_{i<j} \frac{l_i^2-l_j^2}{m_i^2-m_j^2}.
\ee
For example, for $SO(6)$ ($n=3$) and $(r_1,r_2,r_3)=(2,0,0)$, 
we compute $l_1= r_1+ n-1=2+3-1=4, l_2= 1,l_3=0,m_1= 2,m_2=1,m_3=0$, with which (\ref{dimFormula}) yields $\text{dim}\;\Gamma_{[2]}= 20$. Similarly, one can compute  
${\rm dim}\;\Gamma_{[2]}= 35$ for $SO(8)$.

Since we are interested in  $SO(N+1)$, we set  $n=\frac{N+1}{2}$ and write a formula which gives the dimension of irreps for $Q=[2]$ for arbitrary $N$ as
\be\label{dimExample}
{\rm dim} \;\Gamma_{[2]}=N\left(\frac{N+1}{2}+1\right).
\ee
This agrees with the values mentioned above at $N=3,5,7$. Note that although we derived this using the formulas for the even orthogonal case, it turns out that the same expression gives the correct dimensions for the odd orthogonal case (explicit  formulas for the latter are provided in Lecture 24 of \cite{Fulton}).

The quantum dimension is obtained from the classical one in (\ref{dimExample}) by replacing each classical number $x\in\mathbb{Z}$ by the quantum number $\{q^x\}:=q^x-q^{-x}$. By setting $A=q^N$, we obtain the quantum dimension corresponding to the Young diagram $Q=[2]$ to be
\be
d_{[2]}= \frac{\{A\}}{\{q\}} \left(\frac{\{A q\}}{\{q^2\}}+1\right).
\ee
Similarly, we find ${\rm dim}\;\Gamma_{[11]}=N(\frac{N-1}{2}+1)$ and hence $d_{[11]}= \frac{\{A\}}{\{q\}}( \frac{\{A q^{-1}\}}{\{q^2\}}+1)$.
Notably, the dimension $d_{[11]}$ is related to $d_{[2]}$ by 
$q\to -1/q.$

In a similar way, we determine  formulas for the classic dimensions for higher partitions, which are listed in Table~\ref{tab:d_QClassical}.  
Examples with $3$ and $4$ boxes in the Young diagram are
${\rm dim}\;\Gamma_{[3]}=\frac{N(N+1)}{2}(\frac{N+2}{3}+1)$, ${\rm dim}\;\Gamma_{[4]}=\frac{N(N+1)  (N+2)}{6}( 1+ \frac{N+3}{4})$, respectively. 
These yield the quantum dimensions $d_{[3]}=\frac{\{A\}\{A q\}}{\{q\}\{q^2\}}( 1+\frac{\{Aq^2\}}{\{ q^3\}})=(1+\frac{\{q^3\}}{\{ Aq^2\}})S_{[3]}$ and
$d_{[4]}=( 1+\frac{\{q^4\}}{\{A q^3\}})S_{[4]}$, respectively, where $S_Q$ denotes the Schur functions, which are the quantum dimensions for $SU(N)$ and are listed in Table~\ref{table:S_Q}.

\begin{table}[H]
\centering
 \caption{Classical dimensions ${\rm dim} \;\Gamma_{Q}$ of irreps of $SO(N+1)$}.
    \label{tab:d_QClassical}
\scalebox{0.9}[0.9]{
\begin{tabular}{ll}
\toprule
$Q$ & ${\rm dim} \;\Gamma_{Q}$\\
\midrule
$[1]$ &$ N+1$  \\
\hline
\vspace{1mm}$[2]$ &$ \frac{1}{2}N(N+3)$ \\ 
\vspace{1mm}$[11]$ &$ \frac{1}{2}N(N+1)$ \\
\hline
\vspace{1mm}$[3]$ &$\frac{1}{6}N ( N+1) (N+5)$\\
\vspace{1mm}$[21]$ &$\frac{1}{3}(N-1)(N+1)  (N+3)$\\
\vspace{1mm}$[111]$ &$ \frac{1}{6}N (N+1)(N-1)$ \\
\hline
\vspace{1mm}$[4]$ &$ \frac{1}{24}N (N+1)(N+2)(N+7)$\\
\vspace{1mm}$[31]$ &$\frac{1}{8}N(N-1)(N+2) (N+5)$\\
\vspace{1mm}$[22]$ &$ \frac{1}{12}(N-2)(N+1) (N+2)(N+3)$\\
\vspace{1mm}$[211]$ &$\frac{1}{8}N (N+1) (N-2) (N+3)$\\
\vspace{1mm}$[1111]$ &$\frac{1}{24} N (N+1) (N-1) (N-2)$  \\
\hline
\vspace{1mm}$[5]$ &$ \frac{1}{120}N(N+1) (N+2) (N+3) (N+9)$ \\
\vspace{1mm}$[41]$ &$\frac{1}{30}(N-1) N (N+1) (N+3)(N+7)$\\
\vspace{1mm}$[32]$ &$ \frac{1}{24} N(N-2) (N+1) (N+3)(N+5)$ \\ 
\vspace{1mm}$[311]$ &$\frac{1}{20}(N-2)(N-1) (N+1) (N+2)(N+5) $ \\
\vspace{1mm}$[221]$ &$ \frac{1}{24}(N-3)N(N+1)  (N+2) (N+3) $ \\
\vspace{1mm}$[2111]$ &$\frac{1}{30}N (N+1) (N-1) (N-3) (N+3)$ \\
\vspace{1mm}$[11111]$ &$\frac{1}{120}N (N+1)(N-1)(N-2)(N-3)$\\
\bottomrule
\end{tabular}}
\end{table}

An alternative way to obtain  the quantum dimensional formula for $SO(N+1)$ 
 is through the explicit formula (originating from the
quantum version of Weyl's formula \cite{Weyl,Fulton})
\cite{Bouchard,MorozovU}
\ba\label{dimGamma}
d_Q &=& 
\prod_{1\le i\le j\le l(Q)} 
\frac{\{q^{r_i-r_j + j-i}\}\{A q^{r_i+r_j+1-i-j}\}}{\{q^{j-i}\}\{A q^{1-i-j}\}}\nonumber\\
&&\times
\prod_{k=1}^{l(Q)}\biggl( \frac{\{A^{1/2}q^{r_k-k+1/2}\}}{\{A^{1/2} q^{-k+1/2}\}}
\prod_{s=1}^{r_k}\frac{\{A q^{r_k+1-k-s-l(Q)}\}}{\{q^{s-k+l(Q)}\}}\biggr).
\ea
Here $l(Q)$ is the number of non empty rows of $Q$. 
\begin{table}[H]
\vskip 3mm
\centering
\caption{Quantum dimensions for $SU(N)$ with $A=q^N.$}\label{table:S_Q}
\scalebox{0.9}[0.9]{
\begin{tabular}{lcl}
\toprule
$S_{[1]}$ &\hspace{-4mm}=&\hspace{-4mm} $ {\{A\}}/{\{q\}}$   \\
\hline
$ S_{[2]}$ &\hspace{-4mm}=&\hspace{-4mm} $ {\{A\}\{A q\}}/{\{q\}\{q^2\}}$ \\
 $S_{[11]}$ &\hspace{-4mm}=&\hspace{-4mm} ${\{A\}\{A q^{-1}\}}/{\{q\}\{q^2\}}$\\
 \hline
$ S_{[3]}$ &\hspace{-4mm}=&\hspace{-4mm} ${\{A\}\{A q\}\{A q^2\}}/{\{q\}\{q^2\}\{q^3\}}$\\
 $S_{[21]}$ &\hspace{-4mm}=&\hspace{-4mm} ${\{A\}\{A q\}\{A q^{-1}\}}/{\{q\}^2\{q^3\}}$\\
$ S_{[111]}$ &\hspace{-4mm}=&\hspace{-4mm} ${\{A\}\{A q^{-1}\}\{A q^{-2}\}}/{\{q\}\{q^2\}\{q^3\}}$ 
 \\
\hline
$S_{[4]}$ &\hspace{-4mm}=&\hspace{-4mm} ${\{A\}\{A q\}\{A q^2\}\{A q^3\}}/{\{q\}\{q^2\}\{q^3\}\{q^4\}}$ \\
$ S_{[31]}$ &\hspace{-4mm}=&\hspace{-4mm} ${\{A\}\{A q\}\{A q^2\}\{A q^{-1}\}}/{\{q\}^2\{q^2\}\{q^4\}} $\\
 $S_{[22]}$ &\hspace{-4mm}=&\hspace{-4mm} ${\{A\}^2\{A q\}\{A q^{-1}\}}/{\{q\}\{q^2\}^2\{q^3\}} $
 \\   
$S_{[211]}$ &\hspace{-4mm}=&\hspace{-4mm} $ {\{A\}\{A q\}\{A q^{-1}\}\{A q^{-2}\}}/{\{q\}^2\{q^2\}\{q^4\}}$ 
  \\   
$S_{[1111]}$ &\hspace{-4mm}=&\hspace{-4mm} $ {\{A\}\{A q^{-1}\}\{A q^{-2}\}\{A q^{-3}\}}/{\{q\}\{q^2\}\{q^3\}\{q^4\}} $
 \\
\hline
$S_{[5]}$ &\hspace{-4mm}=&\hspace{-4mm} $ {\{A\}\{A q\}\{A q^2\}\{A q^3\}\{A q^4\}}/{\{q\}\{q^2\}\{q^3\}\{q^4\}\{q^5\}}$    \\
 $S_{[41]}$ &\hspace{-4mm}=&\hspace{-4mm} ${\{A q^{-1}\}\{A\}\{Aq\}\{A q^2\}\{A q^{3}\}}/{\{q\}^2\{q^2\}\{q^3\}\{q^5\}}$
\\
 $S_{[311]}$ &\hspace{-4mm}=&\hspace{-4mm} ${
\{A q^{-2}\}\{A q^{-1}\}\{A\}\{A q\}\{A q^2\}}/{\{q\}^2\{q^2\}^2
\{q^5\}}$
\\
 $S_{[32]}$ &\hspace{-4mm}=&\hspace{-4mm} $ {\{A\}^2\{A q\}\{A q^2\}\{A q^{-1}\}}/{\{q\}^2\{q^2\}\{q^3\}\{q^4\}}$ \\
$S_{[221]}$ &\hspace{-4mm}=&\hspace{-4mm} $ {\{A\}^2\{A q^{-1}\}\{A q^{-2}\}\{A q\}}/{\{q\}^2\{q^2\}\{q^3\}\{q^4\}}$
  \\
 $S_{[2111]}$ &\hspace{-4mm}=&\hspace{-4mm} ${\{A q\}\{A\}\{A q^{-1}\}\{A q^{-2}\}\{A q^{-3}\}}/{\{q\}^2\{q^2\}\{q^3\}\{q^5\}}$
 \\
 $S_{[11111]}$ &\hspace{-4mm}=&\hspace{-4mm} ${\{A\}\{A q^{-1}\}\{A q^{-2}\}\{A q^{-3}\}\{A q^{-4}\}}/{\{q\}\{q^2\}\{q^3\}\{q^4\}\{q^5\}}$    \\
\bottomrule
\end{tabular}}
 \vspace{5mm}
 \caption{Quantum dimensions of irreps of $SO(N)$ with $A=q^{N}$.}
\label{tab:d_Q}
\scalebox{0.9}[0.9]{
\begin{tabular}{lcl}
\toprule
 $d_{[1]}$&\hspace{-4mm}=&\hspace{-4mm} $1+    {\{A\}}/{\{q\}} $ \\
\hline
 $d_{[2]}$&\hspace{-4mm}=&\hspace{-4mm} $(1+ {\{q^2\}}/{\{A q\}})S_{[2]}$\\
  $d_{[11]}$&\hspace{-4mm}=&\hspace{-4mm} $(1+ {\{q^2\}}/{\{A q^{-1}\}})S_{[11]}$\\
 \hline
 $d_{[3]}$&\hspace{-4mm}=&\hspace{-4mm} $( 1+ {\{q^3\}}/{\{A q^2\}})S_{[3]}$\\
  $d_{[21]}$&\hspace{-4mm}=&\hspace{-4mm} $( 1+{\{q^3\}}/{\{A\}})S_{[21]} $\\
$d_{[111]}$&\hspace{-4mm}=&\hspace{-4mm} $ (1+{\{q^3\}}/{\{A q^{-2}\}}) S_{[111]}$\\
\hline
 $d_{[4]}$&\hspace{-4mm}=&\hspace{-4mm} $(1+{\{q^4\}}/{\{A q^3\}}) S_{[4]}$\\
 $d_{[31]}$&\hspace{-4mm}=&\hspace{-4mm} $(1+ {\{q^4\}}/{\{A q\}})S_{[31]}$\\
 $d_{[22]}$&\hspace{-4mm}=&\hspace{-4mm} $ {\{A q^2\}\{A q^{-2}\}}/{\{q\}\{q^2\}^2\{q^3\}}(A+q)$\\
 &\hspace{-4mm}&\hspace{-4mm} $\times (1-1/(A q)) (A+ q^3)(1-A^{-1}q^{-3})$\\   
$d_{[211]}$&\hspace{-4mm}=&\hspace{-4mm} $(1+{\{q^4\}}/{\{A q^{-1}\}}) S_{[211]}$ \\   
$d_{[1111]}$&\hspace{-4mm}=&\hspace{-4mm} $ (1+ {\{q^4\}}/{\{A q^{-3}\}}) S_{[1111]}$\\
\hline
 $d_{[5]}$&\hspace{-4mm}=&\hspace{-4mm} $(1+    {\{q^5\}}/{\{A q^4\}})S_{[5]}$  \\
 $d_{[41]}$&\hspace{-4mm}=&\hspace{-4mm} $ (1+ {\{q^5\}}/{\{A q^2\}})S_{[41]}$
\\
  $d_{[311]}$&\hspace{-4mm}=&\hspace{-4mm} $( 1+ {\{q^5\}}/{\{A\}})S_{[311]}$
\\
 $d_{[32]}$&\hspace{-4mm}=&\hspace{-4mm} ${\{A\}\{A q^{-2}\}\{A q\}\{A q^3\}\{A q^5\}}/{\{q\}^2\{q^2\}\{q^3\}\{q^4\}}$\\
 $d_{[221]}$&\hspace{-4mm}=&\hspace{-4mm} ${\{A\}\{A q\}\{A q^2\}\{A q^3\}\{A q^{-3}\}}/{\{q\}^2\{q^2\}\{q^3\}\{q^4\}}$\\
 
 $d_{[2111]}$&\hspace{-4mm}=&\hspace{-4mm} $(1+ {\{q^5\}}/{\{A q^{-2}\}})S_{[2111]}$\\
  $d_{[11111]}$&\hspace{-4mm}=&\hspace{-4mm} $ (1+   {\{q^{5}\}}/{\{A q^{-4}\}})S_{[11111]}$  \\
\bottomrule
\end{tabular}}
\end{table}
For instance, $l(Q)=1$ for  $Q=[1],[2]$ and $[3]$,  and hence the first product  in (\ref{dimGamma}) becomes unity, yielding 
\ba
&&\hspace{-2mm}d_{[1]}= 
\frac{\{A^{1/2}q^{1/2}\}\{A q^{-1}\}}{\{A^{1/2}q^{-1/2}\}\{q\}}=\left(1+ \frac{\{A\}}{\{q\}}\right)\\
&&\hspace{-2mm}d_{[2]}=
\frac{\{A^{1/2}q^{3/2}\}\{A q^{-1}\}\{A\}}{\{A^{1/2}q^{-1/2}\}\{q\}\{q^2\}}=\left(1+ \frac{\{q^2\}}{\{A q\}}\right)\frac{\{A\}\{A q\}}{\{q\}\{q^2\}}\nonumber\\
&&\hspace{-2mm}d_{[3]}=\frac{\{A^{1/2} q^{5/2}\}\{A q^{-1}\}\{A\}\{A q\}\{A q^2\}}{\{A^{1/2}q^{1/2}\}\{A q^4\}\{q\}\{q^2\}\{q^3\}}=\left( 1+\frac{ {\{q^3\}}}{{\{A q^2\}}}\right)\frac{{\{A\}\{A q\}\{A q^2\}}}{{\{q\}\{q^2\}\{q^3\}}}. \nonumber
\ea
These expressions are the same as the ones listed in Table~\ref{tab:d_Q}.

The quantum dimensions $d_Q$ satisfy the following Plethysm identities   
\ba
d_{[1]}^2&=& d_{[2]}+ d_{[11]}+1\nonumber\\
d_{[1]}^3&=& d_{[3]}+2 d_{[21]}+ d_{[111]}+ 3 d_{[1]}\nonumber\\
d_{[1]}^4 &=& d_{[4]}+ 3 d_{[31]}+2d_{[22]}+3 d_{[211]}+d_{[1111]}+ 6d_{[2]}+6 d_{[11]}+ 3\nonumber\\
d_{[2]}d_{[1]}&=& d_{[3]}+ d_{[21]}+d_{[1]}\nonumber\\
d_{[11]} d_{[1]}&=& d_{[21]}+ d_{[111]}+ d_{[1]}\nonumber\\
\hspace{-17mm}d_{[3]} d_{[1]}&=&d_{[4]}+ d_{[31]}+ d_{[2]}\nonumber\\
\hspace{-17mm}d_{[21]} d_{[1]}&=&d_{[31]}+ d_{[22]}+ d_{[211]}+ d_{[2]}+ d_{[11]}\nonumber\\
\hspace{-17mm}d_{[111]} d_{[1]}&=& d_{[211]}+ d_{[1111]} + d_{[11]}.
\ea

\addcontentsline{toc}{section}{Appendix B.  Jaeger's Theorem}
\section*{Appendix B.  Jaeger's Theorem}
The relation between the HOMFLY--PT polynomial and Kauffman polynomials $H=\widehat{KF}$ can be examined through Jaeger's state expansion (defined in p.219 of  \cite{Kauffman}). The skein relation for HOMFLY--PT polynomial is
\be
aH(\img{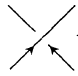})-a^{-1}H(\img{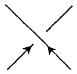})= z H(\img{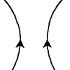})
\ee
and the unnormalised version is defined as $\bar{H}(a,z)=\frac{a-a^{-1}}{z}H(a,z)$. The HOMFLY--PT polynomial has a regular isotopy version $R(\mathcal{K};a,z)=a^{w(\mathcal{K})}\bar{H}(a,z)$, defined by
\ba\label{Rskein}
&&R(\img{positivecrossing.PNG})-R(\img{negativecrossing.PNG})= z R(\img{zero_crossing.PNG});\;\;R(\bigcirc)=\frac{a-a^{-1}}{z}\nonumber\\
&& R(\img{NegLoop.PNG}) = a R(\img{NonLoop.PNG}), \hskip 2mm R(\img{PosLoop.PNG})= a^{-1}R(\img{NonLoop.PNG}).
\ea

According to Jaeger's theorem,  the (unnormalised) Dubrovnik version of the Kauffman polynomial of a knot $\mathcal{K}$ can be expressed as a state sum $[\mathcal{K}](a,q)$ over HOMFLY--PT polynomials 
 \ba\label{jaegerSum}
 [\mathcal{K}](a,q)&=&\sum_{\sigma}[\mathcal{K}|\sigma][\sigma](a,q)
\nonumber\\
&=&\sum_{\sigma}c_\sigma(qa^{-1})^{{\rm rot }(\sigma)}R_{\sigma}(a,q)=\bar{\Lambda}(q^{-1}a^2,q-q^{-1}),
 \ea
 where  $\sigma$ is an oriented diagram (state) and $[\sigma](a,q)=(qa^{-1})^{{\rm rot }(\sigma)}R_{\sigma}(a,q)$, in which $R$ is defined in (\ref{Rskein}) and ${\rm rot}(\sigma)$ is the Whitney degree\footnote{The Whitney degree of a anticlockwise oriented circle is $+1$, while it is $-1$ for clockwise orientation. Given an oriented knot diagram $\sigma$, ${\rm rot}(\sigma)$ can be computed by adding the Whitney degrees of its Seifert circles.}, while the coefficient $[\mathcal{K}|\sigma]=:c_\sigma $ is a polynomial in $\{z,\pm1\}$ as determined by the relation
 \be\label{jaegerskein}
[\img{Kauffnegative.PNG}]=z\left([\img{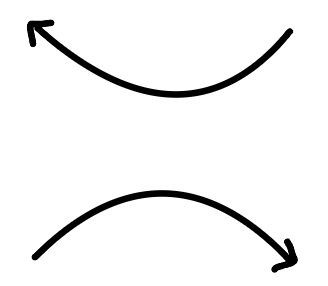}]-[\img{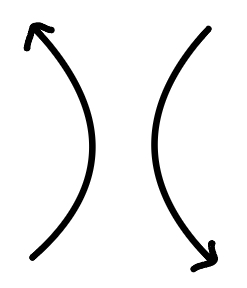}]\right)+[\img{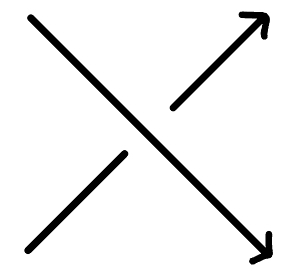}]+[\img{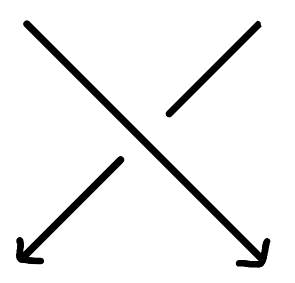}]+[\img{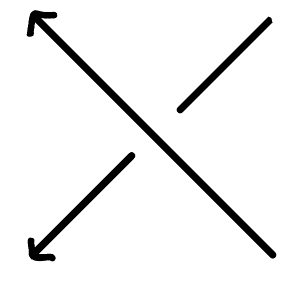}]+[\img{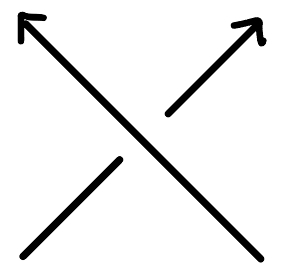}].
 \ee
 A state $\sigma$ contributes in the state sum only if its global orientation is compatible.
 For instance, for the right handed trefoil knot the various oriented diagrams  with compatible orientations that contribute in the state sum of (\ref{jaegerSum}) are shown in Fig.~\ref{fig:TrefStates}. 
 \begin{figure}[h!]
     \centering
\includegraphics[scale=0.17,angle=1]{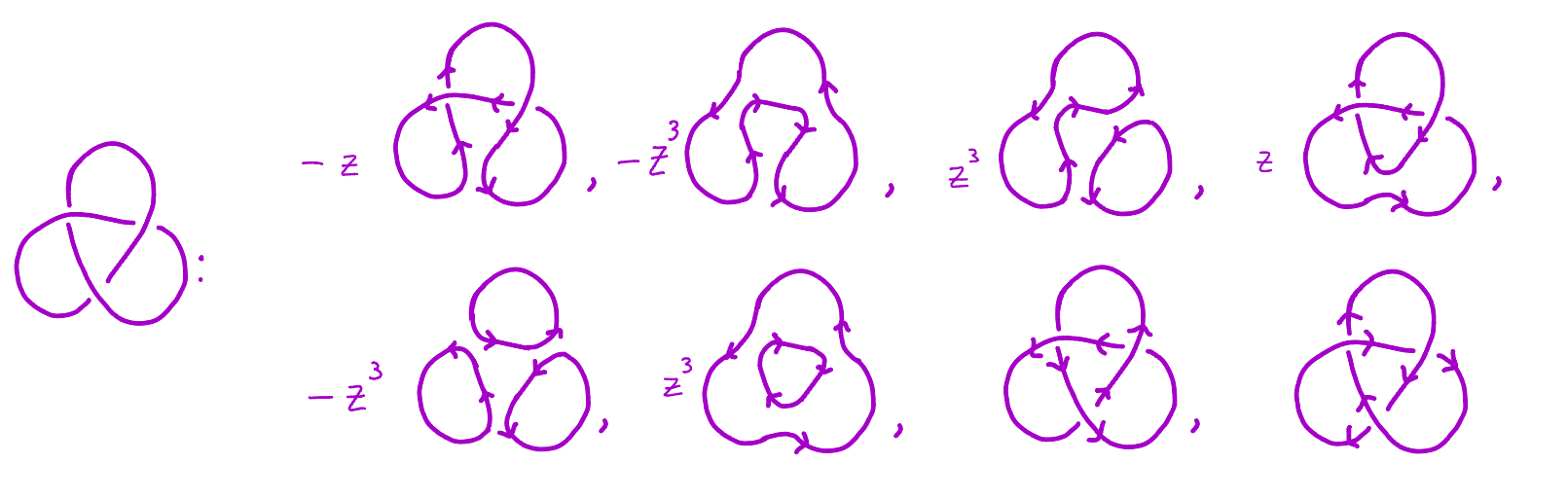}
     \caption{Various  oriented diagrams contributing in the state sum (\ref{jaegerSum}) for the trefoil knot (left), along with their coefficients $[\mathcal{K}|\sigma]$. }
     \label{fig:TrefStates}
 \end{figure}
 Note that each configuration corresponding to diagrams  in the top row will appear 3 times with the same coefficient, since the same splitting  can be applied to any of the 3 crossings of the trefoil (see Fig.~\ref{fig:3options} for an example).
 \begin{figure}[h!]
     \centering
\includegraphics[scale=0.17]{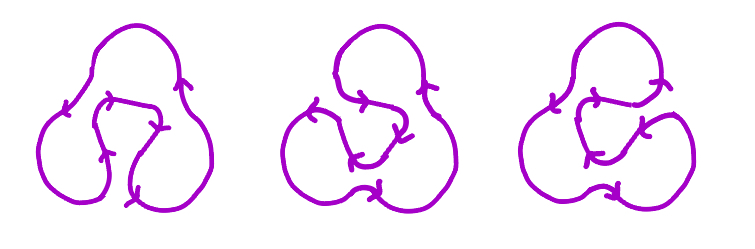}
     \caption{Three oriented states with the same coefficient $c_\sigma(z)=[\mathcal{K}|\sigma]=-z^3$. }
     \label{fig:3options}
 \end{figure}
Hence a total 16 states (out of 27 possibilities) contribute  for the trefoil knot.

 Using  that $\bar{H}(A,z)=A^{-w}R(\mathcal{K};A,z)$ and by setting $A=q^{-1}a^2,$ $z=q-q^{-1}$ so that
$\widetilde{\overline{KF}}(\mathcal{K};q^{-1}a^2,q-q^{-1})=-(q^{-1}a^2)^{-w}[\mathcal{K}](ia,q)$,  (\ref{HKFbar}) can be written as
 \ba
&&\hspace{-12mm}R(q^{-1}a^2,q-q^{-1})=-\frac{1}{2}([\mathcal{K}](ia,q)-[\mathcal{K}](a,-q))\\
\hspace{1mm}&&= -\frac{1}{2}\sum_{\sigma}\left(c_\sigma(q)\left(\frac{q}{ia}\right)^{{\rm rot}(\sigma)}R_{\sigma}(ia,q)-c_\sigma(-q)\left(\frac{-q}{a}\right)^{{\rm rot}(\sigma)}R_{\sigma}(a,-q)\right),\nonumber
 \ea
which becomes a relation  between the HOMFLY--PT polynomial of the knot and its states. This relation is verified to be true in the case of the trefoil knot, for which the sum over states $\sigma$ in the RHS involves the states presented in Fig.~\ref{fig:TrefStates}.

The Jaeger expansion can also be used to write down a formal character expansion for the Dubrovnik polynomial in terms of the characters of $SU(N)$, i.e. the Schur functions $S_Q$, listed in Appendix A. 
By setting $A=q^{-1}a^2$ and $z=q-q^{-1}$, for a knot $\mathcal{K}$ that has a braid representative with $m$ strands and writhe $w$ this can be expressed as
\be\label{JeagerKchi}
\bar{Y}(\mathcal{K};A,z)=(qa^{-2})^w\sum_{\sigma}c_{\sigma}(qa^{-1})^{{\rm rot} (\sigma)}\sum_{Q(\sigma)}h^Q(\sigma)S_Q.
\ee
Here the
last summation refers to  the character expansion corresponding to the state $\sigma$, which usually can be expressed as a braid with less than $m$ strands.
For instance, for the 3-stranded unknot $T(3,1)$ we find
\ba
&&\hspace{-6mm}\bar{Y}(T(3,1);a,q)=(1-q^4a^{-4})S_{[2]}-q^{-2}(1-q^4a^{-4})S_{[11]}\\
&&\hspace{-1mm}+q^3a^{-3}(1+q^2a^{-2})S_{[3]}-qa^{-3}(1+q^2a^{-2})S_{[21]}+q^{-1}a^{-3}(1+q^2a^{-2})S_{[111]},\nonumber
\ea
which simplifies to the standard expression $\bar{Y}(\bigcirc)=\frac{A-A^{-1}}{z}+1$ at $A=q^{-1}a^2$ and $z=q-q^{-1}$, as expected. 

Although  state sums do not provide the most efficient description of the Dubrovnik polynomial and its relation to HOMFLY--PT, as the above examples show, they are still a powerful theoretical tool and hence we include this Appendix here for completeness. An interesting application  could be to use the Jaeger's state sum to derive a homology theory corresponding to the Kauffman polynomial, in an analogous way as the Kauffman bracket (the state sum of the Jones polynomial) led to the Khovanov homology \cite{Khovanov}.  


\end{document}